\documentclass[aps,pra,10pt,tightenlines,longbibliography,notitlepage,amsmath,twocolumn,floatfix,superscriptaddress,eqsecnum]{revtex4-1}
\usepackage[dvipsnames]{xcolor}
\usepackage[colorlinks=true,bookmarks=false,linkcolor=NavyBlue,urlcolor=NavyBlue,citecolor=NavyBlue]{hyperref}
\usepackage{amsmath}
\usepackage{amsfonts}
\usepackage{tabularx}
\usepackage{graphicx}
\usepackage{times}
\usepackage{mathtools}
\usepackage{braket}
\usepackage{amsthm}
\usepackage{amssymb}
\usepackage{enumitem}

\newtheorem{theorem}{Theorem}

\newtheorem{proposition}{Proposition}
\newtheorem{lemma}{Lemma}
\newtheorem{alemma}{Lemma}[section]
\newtheorem{definition}{Definition}[section]
\theoremstyle{remark}
\newtheorem*{remark}{Remark}

\newcommand{\parti}[2]{\frac{\partial #1}{\partial #2}}

\newcommand{\intall}{\int_{-\infty}^{\infty}}
\newcommand{\avg}[1]{\langle#1\rangle}
\newcommand{\Avg}[1]{\left\langle#1\right\rangle}
\newcommand{\abs}[1]{\left|#1\right|}
\newcommand{\bk}[1]{\left(#1\right)}
\newcommand{\Bk}[1]{\left[#1\right]}
\newcommand{\BK}[1]{\left\{#1\right\}}
\newcommand{\bko}[1]{\lbrack #1 \rparen}
\newcommand{\expect}{\mathbb E}
\newcommand{\var}{\mathbb V}

\newcommand{\floor}[1]{\left\lfloor#1\right\rfloor}
\newcommand{\ceil}[1]{\left\lceil#1\right\rceil}

\newcommand{\norm}[1]{\lVert #1 \rVert}

\newcommand{\Norm}[1]{\left\lVert #1 \right\rVert}
\newcommand{\opnorm}[1]{\norm{#1}_{\mathrm{op}}}

\newcommand{\fnorm}[1]{\norm{#1}_{\mathrm{HS}}}
\newcommand{\deff}{\delta_{\mathrm{eff}}}

\DeclareMathOperator{\spn}{span}
\DeclareMathOperator{\cspn}{\overline{span}}
\DeclareMathOperator{\trace}{tr}

\DeclareMathOperator{\supp}{supp}
\newcommand{\el}{u}
\newcommand{\elalt}{v}
\newcommand{\chol}{L}

\newcommand{\fig}[3]{
\begin{figure}[htbp!]
\centerline{\includegraphics[width=#1\textwidth]{#2}}
\caption{\label{#2}#3}
\end{figure}
}
\begin{document}

\title{Quantum limit to subdiffraction incoherent optical imaging.
II. A parametric-submodel approach}

\author{Mankei Tsang}
\email{mankei@nus.edu.sg}
\homepage{https://blog.nus.edu.sg/mankei/}
\affiliation{Department of Electrical and Computer Engineering,
  National University of Singapore, 4 Engineering Drive 3, Singapore
  117583}

\affiliation{Department of Physics, National University of Singapore,
  2 Science Drive 3, Singapore 117551}

\date{\today}


\begin{abstract}
  In a previous paper [M.~Tsang, Phys.~Rev.~A \textbf{99}, 012305
  (2019)], I proposed a quantum limit to the estimation of object
  moments in subdiffraction incoherent optical imaging. In this
  sequel, I prove the quantum limit rigorously by infinite-dimensional
  analysis.  A key to the proof is the choice of an unfavorable
  parametric submodel to give a bound for the semiparametric
  problem. By generalizing the quantum limit for a larger class of
  moments, I also prove that the measurement method of spatial-mode
  demultiplexing (SPADE) with just one or two modes is able to achieve
  the quantum limit. For comparison, I derive a classical bound for
  direct imaging using the parametric-submodel approach, which
  suggests that direct imaging is substantially inferior.
\end{abstract}

\maketitle

\section{Introduction}
No problem is more essential in optics than the resolution limit of
incoherent imaging \cite{villiers}. Its importance to astronomy
\cite{goodman_stat}, fluorescence microscopy \cite{pawley}, and
countless other imaging applications can hardly be overestimated. The
fact that diffraction and photon shot noise play dual roles in
limiting the ultimate resolution suggests that quantum information
theory sets the right foundation for the problem
\cite{helstrom,tsang19a}. That said, the mathematics of quantum
information is daunting, its application to imaging more so, and
progress has been mostly limited to toy examples, in which only a few
parameters or hypotheses about the object are assumed to be unknown
\cite{helstrom,tsang19a}.

In recent years, the surprising results concerning two point sources
in Ref.~\cite{tnl} have triggered renewed interest in the
quantum-information perspective \cite{tsang19a}, as well as research
efforts towards more general cases
\cite{yang16,tsang17,tsang18a,tsang19b,zhou19,tsang19,bonsma19,liang21a,dutton19,prasad20,bisketzi19,fiderer21,bearne21,bao21,pushkina21,matlin21,grace21}.
In terms of computing the quantum limits, some researchers attack the
case of a few point sources and many parameters with numerical methods
\cite{bisketzi19,fiderer21}, but the harsh computational demands mean
that alternative approaches are necessary to deal with more complex
objects and high-dimensional parameter spaces. In this regard,
Refs.~\cite{yang16,tsang17,tsang18a,tsang19,tsang19b,zhou19,bonsma19,liang21a}
are able to make progress by framing the problem as object-moment
estimation, while assuming little about the object distribution. The
prequel of this paper, in particular, proposes a quantum limit to
moment estimation in the form of a quantum Cram\'er-Rao bound
\cite{tsang19}.

Two mathematical issues arise in the moment estimation problem: the
infinite dimensionality of the parameter space, since an extended
object may depend on infinitely many scalar parameters, and the
infinite dimensionality of the quantum states, since an extended
object may excite infinitely many spatial modes. The prequel sweeps
these issues ``under the rug'' and relies on finite-dimensional
arguments.  The main goal of this paper is to prove the quantum limit
rigorously by infinite-dimensional analysis
\cite{rudin_baby,debnath05,dieudonne,flett,bogachev,reed_simon,davies,parthasarathy05,parthasarathy67,parthasarathy92,holevo11,holevo}.

The theory of quantum semiparametric estimation, recently proposed in
Ref.~\cite{tsang20}, is a key to the proof. Although the parameter
space is infinite-dimensional, the theory enables one to derive lower
error bounds by considering parametric submodels, each of which
depends on just a scalar parameter and is much easier to handle.

The second goal of this paper is to prove the quantum optimality of a
measurement method called spatial-mode demultiplexing (SPADE)
\cite{tnl,tsang19a} for unbiased moment estimation. Previous works are
fixated on the estimation of simple moments (in the form of
$\int x^\mu P(dx)$ for an object distribution $P$ and a positive
integer $\mu$). To estimate a simple moment without bias, measurements
of infinitely many spatial modes are needed, and it becomes difficult
to even prove that an unbiased estimator exists
\cite{tsang17,tsang18a,tsang19,tsang19b}. In practice, of course, only
a finite number of modes can be measured \cite{fabre20,boucher20}, so
the existing results do not reflect well on the optimality of SPADE in
practice. This paper generalizes the quantum limit for a larger class
of moments, so that, with just one or two modes, SPADE can still
achieve the quantum limit for the generalized moments it is naturally
measuring. There does not seem to be any compelling reason in practice
to prefer a simple moment over a generalized moment pretty close to
it, so there is, arguably, little loss of practical relevance and much
to gain in the rigor by generalizing the moments.

The final goal of this paper is to give a bound for generalized moment
estimation with direct imaging, thus proving the superiority of SPADE.
While a similar result has been proposed in
Refs.~\cite{tsang17,tsang18a,tsang19b}, it relies on a special
assumption about the point-spread function that is hard to check and
has been verified only for a Gaussian point-spread function. I attempt
to relax the assumption by appealing again to the parametric-submodel
approach. Although the result is still not as general as one would
like, it at least establishes conditions that are easier to check and
paves the way for further generalizations.

This paper is organized as follows. Section~\ref{sec_model} introduces
a quantum model of incoherent optical imaging. Section~\ref{sec_semi}
introduces the quantum semiparametric estimation theory.
Section~\ref{sec_submodel} proposes a parametric submodel for the
derivation of the quantum limit.  Section~\ref{sec_qbound} presents
the quantum limit. Section~\ref{sec_spade} proves that SPADE with one
or two modes can still achieve the quantum limit.
Section~\ref{sec_direct} gives a bound for direct imaging, and
Sec.~\ref{conclusion} is the conclusion. The appendices contain the
more technical proofs and remarks.

\section{\label{sec_model}Model}
Let $\mathcal H_0$ be a 1-dimensional (1D) Hilbert space for the
vacuum, $\mathcal H$ be a separable Hilbert space that models the
spatial modes of light, $\tau_0$ be the vacuum state on
$\mathcal H_0$, and $\tau$ be a one-photon state on $\mathcal H$.
Suppose that the state in each temporal mode can be modeled as
\begin{align}
\rho &= (1-\epsilon) \tau_0 \oplus \epsilon\tau,
\label{rho}
\end{align}
where $\oplus$ denotes the direct sum and $\epsilon$ is the one-photon
probability per temporal mode. With $M$ temporal modes, the state is
assumed to be the tensor power $\rho^{\otimes M}$, and the expected
photon number in all modes is
\begin{align}
N &\equiv M\epsilon.
\label{N}
\end{align}
The validity of this ``rare-photon'' model with $\epsilon \ll 1$ to
describe thermal light at optical frequencies is studied extensively
in Refs.~\cite{tnl,tsang19,tsang19a}. Taking the limit
$\epsilon \to 0$ while keeping $N$ fixed leads to a Poisson model
\cite{tsang21} that agrees with semiclassical optics
\cite{goodman_stat}, although it is not necessary to consider the
Poisson limit in the following.

Assume 1D imaging for simplicity.  Let the set of object-plane
coordinates be $\mathbb R$ and $\Sigma$ be the Borel sigma-algebra of
$\mathbb R$ \cite{parthasarathy05}. For an object that emits spatially
incoherent light and is imaged with a diffraction-limited system,
$\tau$ can be modeled as \cite{tnl,tsang19}
\begin{align}
\tau(P) &= \int  e^{-i\hat kx}\ket{\psi}\bra{\psi} e^{i\hat kx}P(dx),
\label{tau}
\end{align}
where $P:\Sigma \to [0,1]$ is a probability measure that
models the object distribution normalized by the total brightness,
$\ket{\psi} \in \mathcal H$ with $\braket{\psi|\psi} = 1$ models the
coherent point-spread function of the imaging system, and $\hat k$ is
a self-adjoint operator on $\mathcal H$ for the optical spatial
frequency. I call any self-adjoint operator an observable in the
following. Figure~\ref{imaging_setup} illustrates the imaging system.

\fig{0.48}{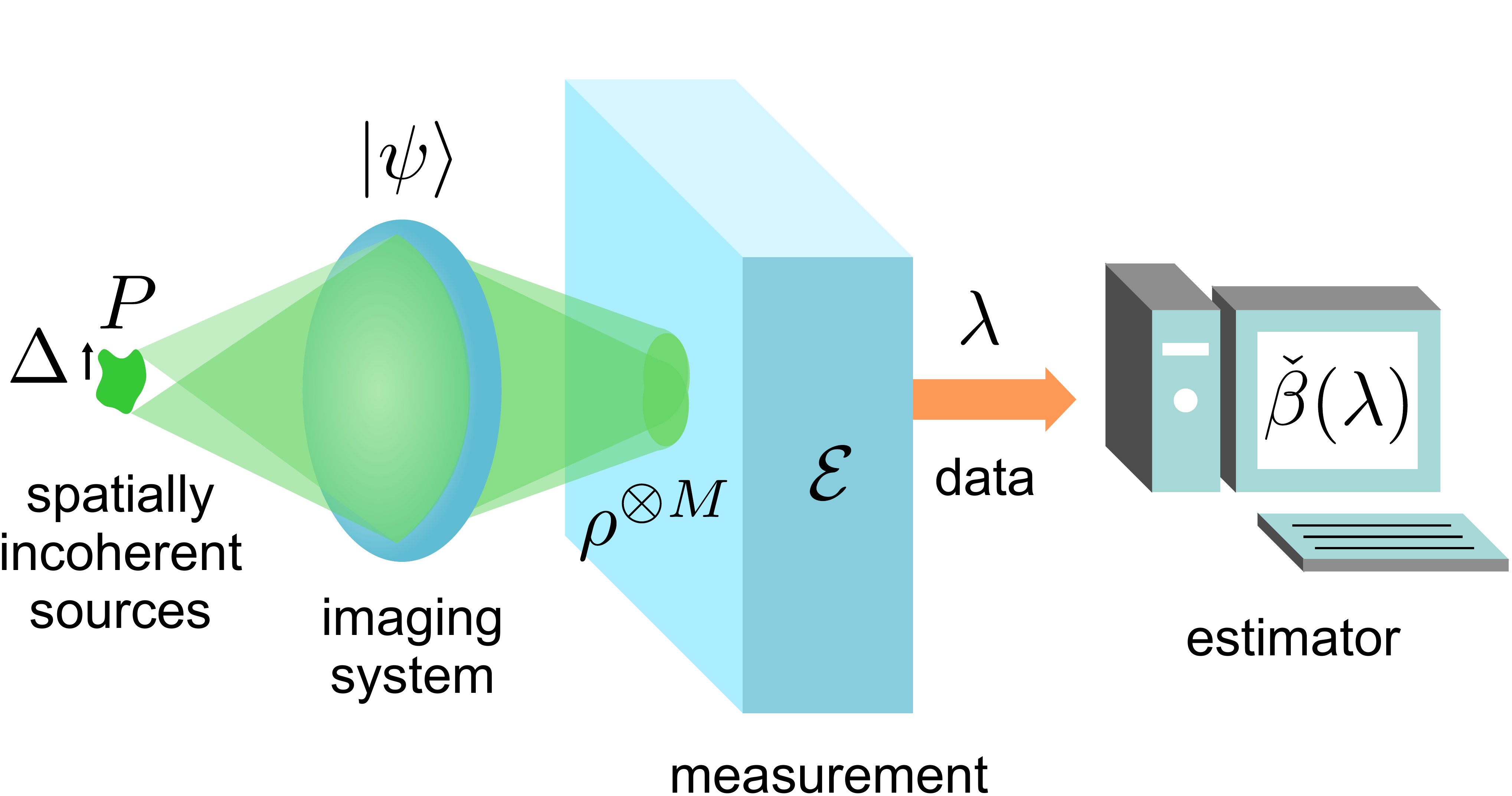}{A schematic of an imaging system.  See the
  main text for the definitions of the symbols.}

To work with the probability space
$(\mathbb R,\Sigma,P)$, it will be useful to define an
inner product, weighted by $P$, between two real functions
$\el,\elalt:\mathbb R \to \mathbb R$ as
\begin{align}
\Avg{\el,\elalt}_P &\equiv \int   \el(x)\elalt(x)P(dx),
\label{inner_P}
\end{align}
the corresponding norm as
\begin{align}
\norm{\el}_P &\equiv \sqrt{\avg{\el,\el}_P},
\end{align}
and the resulting real Hilbert space as \cite{parthasarathy05}
\begin{align}
L_2(P) &\equiv \BK{\el: \norm{\el}_P < \infty}.
\end{align}
I assume $L_2(P)$ to be separable in the following.

$\hat k$ and $\ket{\psi}$ lead to another probability space.  Let the
spectral representation of $\hat k$ be \cite{reed_simon}
\begin{align}
\hat k &= \int k E(dk),
\label{k_spectral}
\end{align}
where $E$ is a projection-valued measure on $(\mathbb
R,\Sigma)$. Define the spatial-frequency measure with respect to
$\hat k$ and $\ket{\psi}$ as
\begin{align}
Q(\cdot) &\equiv \bra{\psi}E(\cdot)\ket{\psi}.
\label{Q}
\end{align}
For example, if $\hat k$ is a continuous variable, then
$E(dk) = \ket{k}\bra{k} dk$ with $\braket{k|k'} = \delta(k-k')$ in the
Dirac notations, and $Q(dk) = |\braket{k|\psi}|^2 dk$, where
$\braket{k|\psi}$ is the optical transfer function of the imaging
system. 

In quantum information theory \cite{hayashi,demkowicz15}, it is often
useful to find a purification of $\tau$ in a larger Hilbert space
$\mathcal H \otimes \mathcal H'$, such that
\begin{align}
\tau &= \trace' \Psi,
&
\Psi &\equiv \ket\Psi\bra\Psi,
&
\ket{\Psi} &\in \mathcal H\otimes\mathcal H',
\label{pur}
\end{align}
where $\trace'$ denotes the partial trace over $\mathcal H'$. A
natural choice is to take $\mathcal H' = L_2(P)$ using the
representation
\begin{align}
1 \in L_2(P) &\leftrightarrow \ket{\phi} \in \mathcal H',
\\
\el(x) \in L_2(P) &\leftrightarrow \el(\hat x)\ket{\phi} \in \mathcal H',
\end{align}
where $\hat x$ is called the canonical multiplication operator, which
is self-adjoint \cite[Exercise~12.7(iii)]{parthasarathy92}, and
$\ket{\phi}$ is called a cyclic vector, satisfying
$\braket{\phi|\phi} = 1$ \cite{reed_simon}. To adhere to physics
terminology, I call $\ket{\phi}$ a purification of $P$. The
purification of $\tau$ then becomes
\begin{align}
\ket{\Psi} &= e^{-i\hat k\otimes \hat x}\ket{\psi}\otimes \ket{\phi}.
\label{natural_pur}
\end{align}
Let the spectral representation of $\hat x$ be
\begin{align}
\hat x &= \int x E'(dx),
\end{align}
where $E'$ is a projection-valued measure on
$(\mathbb R,\Sigma)$. Then
\begin{align}
P(\cdot) &= \bra{\phi} E'(\cdot)\ket{\phi}.
\label{spectral}
\end{align}
For example, if $P$ can be expressed in terms of a density $f(x)$ with
respect to the Lebesgue measure, then $\ket{\phi}$ should satisfy, in
the Dirac notations,
$P(dx) = f(x) dx = \bra{\phi}E'(dx) \ket{\phi} =
\abs{\braket{x|\phi}}^2 dx$. In other words, given an $f$, the
purification should have a wavefunction that satisfies
$|\braket{x|\phi}| = \sqrt{f(x)}$.

$\tau$ can be purified in other ways by applying isometries on
$\mathcal H'$ to $\ket{\Psi}$. Let $U:\mathcal H' \to \mathcal H''$ be
an isometry that satisfies $U^\dagger U = I_{\mathcal H'}$,
$I_{\mathcal H'}$ being the identity operator on $\mathcal H'$.  Then
an alternative purification is
\begin{align}
\ket{\Phi} &= I_{\mathcal H}\otimes U\ket{\Psi},
&
\tau &= \trace'' \Phi,
&
\Phi &\equiv \ket{\Phi}\bra{\Phi}.
\label{alt_pur}
\end{align}
A fruitful choice made in Ref.~\cite{tsang19} is as follows. 

\begin{lemma}
\label{lem_pur}
Consider the state given by Eq.~(\ref{tau}). Assume that the support
of $P$ \cite{parthasarathy67}, denoted by $\supp P$, is infinite but
bounded, viz.,
\begin{align}
\#\supp P &= \infty,
\\
\Delta \equiv \sup_{x \in \supp P} |x| &\in (0,\infty).
\label{Delta}
\end{align}
Assume also that $\hat k$ is bounded.  Then a purification in
$\mathcal H\otimes\mathcal H''$ is given by
\begin{align}
\ket{\Phi} &= 
\sum_{p=0}^\infty  \sum_{n=0}^p  \frac{(-i\hat k)^p}{p!}\chol_{pn}
 \ket{\psi}\otimes\ket{n},
\label{Phi}
\end{align}
where $\{\ket{n}:n\in\mathbb N_0\}$ is an orthonormal sequence in
$\mathcal H''$ and $\chol$ is the lower-triangular matrix obtained by
applying the Cholesky factorization algorithm
\cite[Algorithm~A.1]{sarkka} to the Hankel matrix
\begin{align}
H_{pq} &\equiv \Avg{x^p,x^q}_P,
\quad
p,q \in \mathbb N_0.
\label{hankel}
\end{align}
The infinite series in Eq.~(\ref{Phi}) converges strongly
\cite[Definition~3.3.9]{debnath05}.
\end{lemma}
The proof is deferred to Appendix~\ref{app_lem_pur}.


The assumptions about $P$ ensure that $H$ and $\chol$ are defined. The
bounded $\hat k$ is a technical assumption to ensure the convergence
of results.  Physically, the $\#\supp P = \infty$ assumption means
that the object consists of infinitely many point sources, or in
other words, it is modeled as an extended object
\cite{goodman_stat}. The bounded support simply means that the object
has a finite size, while a bounded $\hat k$ simply means that the
imaging system has a finite bandwidth.


\section{\label{sec_semi}Semiparametric estimation}
Before proceeding further with the imaging problem, I review the
quantum semiparametric theory proposed in Ref.~\cite{tsang20}, which
is necessary to deal with the infinite-dimensional parameter space of
the problem.

Let the statistical model of quantum states on a Hilbert space
$\mathcal H$ be $\mathbf G \equiv \{\rho(g):g \in \mathcal G\}$, where
$\mathcal G$ is a possibly infinite-dimensional parameter space, and
let the parameter of interest be a scalar
$\beta:\mathcal G \to \mathbb R$. Let $\rho \in \mathbf G$ be the true
state. Define an inner product, weighted by $\rho$, between two
observables $\el$ and $\elalt$ as
\begin{align}
\Avg{\el,\elalt}_\rho &\equiv \trace \bk{\el\circ \elalt}\rho,
\end{align}
where $\el\circ \elalt \equiv (\el\elalt+\elalt\el)/2$ denotes the Jordan
product.  The corresponding norm is
\begin{align}
\norm{\el}_\rho &\equiv \sqrt{\avg{\el,\el}_\rho}.
\end{align}
With respect to this inner product, define the real Hilbert space
$L_2(\rho)$ as the completion of the set $\mathcal B(\mathcal H)$ of
bounded observables \cite{holevo11,holevo}.  Within $L_2(\rho)$,
define a subspace of zero-mean observables as
\begin{align}
  \mathcal Z &\equiv \BK{\el \in L_2(\rho): \Avg{I_{\mathcal H},\el}_\rho = 0}.
\end{align}
Consider a 1D submodel, containing the true $\rho$, given by
\begin{align}
\BK{\sigma(\theta):\theta \in \mathcal R \subseteq\mathbb R,
\sigma(\theta_0) = \rho} \subseteq \mathbf G,
\end{align}
where $\mathcal R$ is an open interval containing $\theta_0$.  Let the
overdot denote the derivative with respect to the parameter that is
evaluated at the truth, such as
\begin{align}
\dot\sigma &\equiv
\left.\parti{\sigma(\theta)}{\theta}\right|_{\theta=\theta_0}.
\end{align}
Assume that the submodel is regular, as defined by the conditions that
$\sigma(\theta)$ is differentiable at the truth (such that
$\dot\sigma$ is trace-class \cite{holevo11}) and
$|\trace \el \dot\sigma| \le C \norm{\el}_\rho$ for all
$\el \in \mathcal B(\mathcal H)$ and some constant $C$. Then the
bounded linear functional $\trace \el \dot\sigma$ of $\el$ is
continuous with respect to the norm $\norm{\el}_\rho$
\cite[Theorem~1.5.7]{debnath05} and can be extended uniquely to be
defined on $L_2(\rho)$ \cite[Theorem~1.5.10]{debnath05}. By the Riesz
representation theorem \cite[Theorem~3.7.7]{debnath05}, there exists a
unique $S^\sigma \in L_2(\rho)$, called the score, such that
\begin{align}
\trace \el \dot\sigma &= \Avg{\el,S^\sigma}_\rho 
\quad \forall \el \in L_2(\rho).
\label{score}
\end{align}
This abstract definition of the score is due to Holevo
\cite{holevo11,holevo}. To put it another way, suppose that an
observable $\hat S^\sigma$, called a symmetric logarithmic derivative,
is a solution to the Lyapunov equation
\begin{align}
\dot\sigma &= \hat S^\sigma \circ \rho,
\end{align}
and $\norm{\hat S^\sigma}_\rho < \infty$.  Then
$\trace \el \dot\sigma = \avg{\el,\hat S^\sigma}_\rho$ for all
$\el \in \mathcal B(\mathcal H)$ by Eq.~(2.8.88) in
Ref.~\cite{holevo11}, and $\hat S^\sigma$ must be in the equivalence
class of $S^\sigma$ by the uniqueness of $S^\sigma$.  In other words,
a regular submodel is defined by having a finite Helstrom information
$\norm{\hat S^\sigma}_\rho^2 = \norm{S^\sigma}_\rho^2$. With all that
said, it is unimportant to make the distinction between the element
and its observables in the following.

Let $\mathcal S$ be the set of all regular 1D submodels of
$\mathbf G$. Let $\{S\}$ be the tangent set, defined as the set of the
scores of all such submodels, viz.,
\begin{align}
\BK{S} \equiv \BK{S^\sigma: \sigma \in \mathcal S} \subseteq \mathcal Z,
\end{align}
and let the tangent space $\mathcal T$ be the closed linear span of
the tangent set, viz.,
\begin{align}
\mathcal T \equiv \cspn\BK{S} \subseteq \mathcal Z.
\end{align}
Define also the set of influence observables with respect to $\beta$
as
\begin{align}
\mathcal D &\equiv 
\BK{\delta \in \mathcal Z: \Avg{\delta,S^\sigma}_\rho = \dot\beta^\sigma
\ \forall \sigma \in \mathcal S},
\end{align}
where $\dot\beta^\sigma$ is obtained by expressing $\beta$ as a
function of $\theta$ for the submodel $\{\sigma(\theta)\}$ and taking
the derivative. Let $\mathcal E$ be the positive operator-valued
measure that models a measurement. The mean-square error of an
unbiased estimator $\check\beta$ at the truth becomes
\begin{align}
\mathsf E &\equiv \int \Bk{\check\beta(\lambda)-\beta}^2 
\trace \mathcal E(d\lambda)\rho.
\end{align}
If $\mathcal D$ is not empty, the generalized Helstrom bound on
$\mathsf E$ for any measurement and any unbiased estimator is given by
\begin{align}
\mathsf E &\ge \tilde{\mathsf H} = \norm{\deff}_\rho^2,
\end{align}
where $\deff$ is the efficient influence given by the
orthogonal projection of any $\delta\in\mathcal D$ into $\mathcal T$,
denoted by
\begin{align}
\deff &\equiv \Pi(\delta|\mathcal T)
\quad
\forall \delta \in \mathcal D.
\end{align}
The result of the projection is unique, as $\mathcal D$ is the affine
subspace $\deff + \mathcal T^\perp$, where
$\mathcal T^\perp$ is the orthocomplement of $\mathcal T$ in
$\mathcal Z$. Figure~\ref{tangent_space} illustrates the essential
concepts from the geometric perspective.

\fig{0.4}{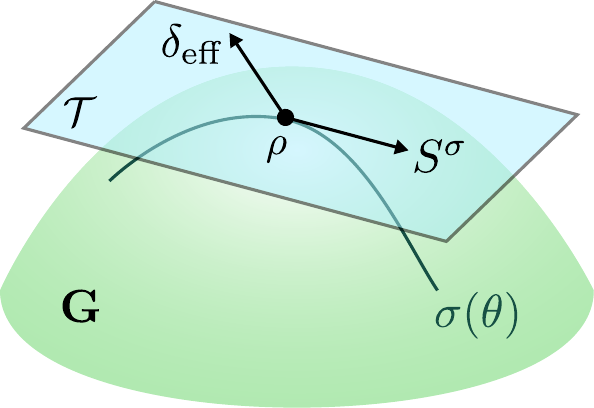}{A geometric picture of the family of states
  $\mathbf G$ as a manifold, the true state $\rho$ as a point, a 1D
  submodel $\sigma(\theta)$ as a line, its score $S^\sigma$ as a
  tangent vector, $\mathcal T$ as the tangent space, and the efficient
  influence $\deff$ as another vector, which is the
  gradient of $\beta$ in $\mathcal T$.}

As discussed in Ref.~\cite{tsang20}, the tangent space for the imaging
problem turns out to be nontrivial and a closed-form solution for the
generalized Helstrom bound is difficult to obtain, so I resort to
looser bounds in the following. A useful lemma, generalizing the
classical Lemma~25.19 in Ref.~\cite{vaart} and a similar result for
the finite-dimensional quantum case in Ref.~\cite{gross20}, is as
follows.
\begin{lemma}
\label{lem_submodel}
\begin{align}
\tilde{\mathsf H} &= \sup_{\el \in \spn\{S\}} \mathsf H(\el),
\label{supremum}
\\
\mathsf H(\el) &\equiv \norm{\Pi(\delta|\spn \el)}_\rho^2
\quad
\forall \delta \in \mathcal D,
\end{align}
where $\spn \el$ is the Hilbert space spanned by one element $\el$.  If
$\el \neq 0$,
\begin{align}
\mathsf H(\el) &= 
\frac{\avg{\delta,\el}_\rho^2}{\norm{\el}_\rho^2}.
\label{H_formula}
\end{align}
\end{lemma}

The proof is deferred to Appendix~\ref{app_lem_submodel}.

If the element $\el$ in $\spn\{S\}$ is the score of a 1D submodel, as is
often the case, $\mathsf H(\el)$ is the Helstrom bound and
$\avg{\delta,\el}_\rho = \dot\beta$ with respect to the
submodel. Lemma~\ref{lem_submodel} then implies that one may consider
only 1D submodels in the evaluation or bounding of
$\tilde{\mathsf H}$. To obtain a tight lower bound on
$\tilde{\mathsf H}$ via $\mathsf H(\el)$, the submodel should be made as
unfavorable to the estimation as possible.

\section{\label{sec_submodel}An unfavorable parametric submodel}
For the imaging problem, let
\begin{align}
\mathbf G &= \BK{\Bk{\rho(P)}^{\otimes M}: P \in \mathcal P},
\label{G}
  \\
\mathcal P &= 
\textrm{all probability measures on }(\mathbb R,\Sigma),
\label{set}
\end{align}
where $\rho(P)$ is given by Eqs.~(\ref{rho})--(\ref{tau}).
Equation~(\ref{set}) implies that no knowledge about the object is
assumed, other than the fact that it is spatially incoherent and the
expected total photon number $N$ is known. If $N$ is unknown, a
submodel with a fixed $N$ still gives a valid lower bound on
$\tilde{\mathsf H}$, by Lemma~\ref{lem_submodel}. Let the parameter of
interest be the linear functional
\begin{align}
\beta(P) &= \int  b(x)P(dx)
\label{beta}
\end{align}
for a given real function $b(x)$.  If $b(x)$ is a monomial, viz.,
\begin{align}
b(x) &= x^\mu, \quad \mu \in \mathbb N_1,
\label{monomial}
\end{align}
I call the $\beta$ a simple moment, as considered in previous works
\cite{tsang17,tsang18a,tsang19,tsang19b,zhou19,bonsma19,liang21a}. If
$\beta$ is not necessarily simple, I call it a generalized moment.

Let $P_0 \in \mathcal P$ be the true measure and consider the submodel
\begin{align}
P_\theta(\el) &= \int_{\el}  g_\theta(x)P_0(dx),
\quad
\el \in \Sigma,
\label{P_theta}
\\
g_\theta(x) &\equiv \frac{1 + \tanh[\theta S(x)]}
{\int \{1 + \tanh[\theta S(x)]\}P_0(dx)},
\label{g}
\\
\theta &\in (-c,c), \quad 0 < c < \infty,
\label{theta}
\end{align}
where the Radon-Nikodym derivative $g_\theta(x)$ given by
Eq.~(\ref{g}) is chosen for its convenient properties \cite{bickel93}
(see also Lemma~\ref{lem_g} in Appendix~\ref{app_thm_qbound}) and
$S(x)$ is assumed to be a function in $L_2(P_0)$ with zero mean at the
truth, viz.,
\begin{align}
\int  S(x) P_0(dx) &= \Avg{1,S}_{P_0} = 0.
\end{align}
Each $P_\theta$ is a valid probability measure, and
$\{P_\theta:\theta \in (-c,c)\}$ also contains the true $P_0$ at
$\theta = 0$, so $\{P_\theta\}$ is a valid submodel for the purpose of
Lemma~\ref{lem_submodel}.

It is straightforward to show that $S(x)$ is the classical score of
the submodel $\{P_\theta\}$.  The quantum score of the corresponding
submodel $\{\tau(P_\theta):\theta \in (-c,c)\}$ for each photon is the
pushforward of $S$ by the map $\tau$ \footnote{$\tau_*$ is called the
  score operator in the classical statistics literature
  \cite{bickel93,vaart}, but that term would cause confusion with the
  quantum operators so I do not use it.}, denoted by $\tau_* S$, to
borrow the terminology from differential geometry
\cite{lee03}. Figure~\ref{pushforward} illustrates the concept from
the geometric perspective.  Formally, the observables in the
equivalence class of $\tau_* S$ obey the Lyapunov equation
\begin{align}
(\tau_* S) \circ \tau(P_0) &= \dot\tau = \tau(\dot P) = \tau(S P_0)
\\
&= \int  e^{-i\hat k x}\ket{\psi}\bra{\psi}e^{i\hat k x} S(x) P_0(dx).
\end{align}
The Helstrom bound for the $M$-temporal-mode submodel
$\{[\rho(P_\theta)]^{\otimes M}:\theta \in (-c,c)\}$ becomes
\begin{align}
\mathsf H &= \frac{\dot\beta^2}{N\norm{\tau_* S}_{\tau(P_0)}^2},
\label{helstrom_submodel}
\\
\dot\beta &= \int  b(x) S(x) P_0(dx) = \Avg{b,S}_{P_0}.
\label{inner_beta}
\end{align}
By Lemma~\ref{lem_submodel}, Eq.~(\ref{helstrom_submodel}) is a lower
bound on the generalized Helstrom bound $\tilde{\mathsf H}$
for the semiparametric problem.

\fig{0.48}{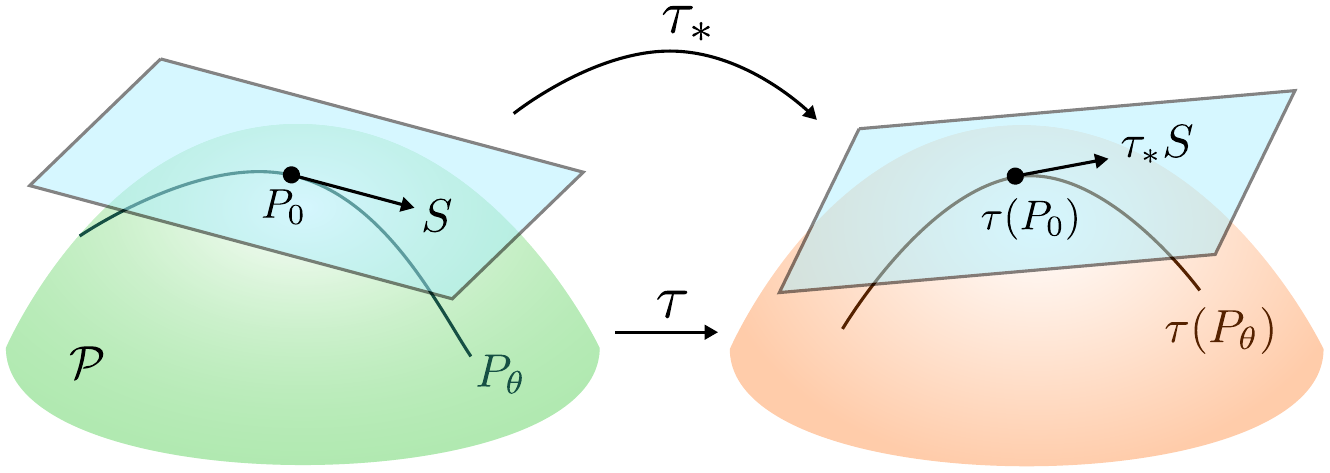}{$\tau$ is a map from a set of input
  probability measures to a set of output quantum states.  $\tau_*$ is
  the pushforward of each score from the input tangent space to the
  output tangent space.}

The prequel of this paper assumes a 1D model with $\theta = \beta$
being a simple moment of order $\mu$ while all the other moments are
fixed \cite{tsang19}. It justifies the quantum bound by appealing to
the inequalities
$\mathsf E \ge (J^{-1})_{\mu\mu} \ge 1/J_{\mu\mu} \ge 1/K_{\mu\mu}$
for the Fisher information matrix $J$ and the Helstrom information
matrix $K$. But it is unclear whether such a model is a valid
submodel, and whether the inequalities are justified for the
infinite-dimensional semiparametric model. 
  Reference~\cite{zhou19}, a related work that studies the classical
  Fisher information for the same moment estimation problem, shares
  the same issues.  Here I am able to alleviate these doubts by
explicitly constructing a valid submodel and appealing to
Lemma~\ref{lem_submodel}.

Note that multiplying $S$ by any nonzero constant does not change
$\mathsf H$, as $\tau_* S$ and $\avg{b,S}_{P_0}$ are both linear with
respect to $S$, so there is no loss of generality if $S$ is normalized
to
\begin{align}
\norm{S}_{P_0} &= 1.
\end{align}
To make the submodel bound $\mathsf H$ a tight lower bound on the
semiparametric $\tilde{\mathsf H}$, there are two heuristic
considerations in the choice of $S$:
\begin{enumerate}
\item The $\tau$ map should shrink its norm as much as possible.

\item $\dot\beta = \avg{b,S}_{P_0} \neq 0$.
\end{enumerate}
Since $\tau$ models the imaging process with a bandwidth limit,
physical intuition suggests that picking a highly oscillatory function
for $S$ may lead to significant norm shrinkage. A convenient choice in
this regard is an orthonormal polynomial $a_n(x)$ specified by
Lemma~\ref{lem_poly}, defined here with respect to the true measure
$P_0$.  It is already normalized, and each $a_n(x)$ is more
oscillatory for higher $n$, as each $a_n(x)$ has $n$ zeros within the
support of $P_0$ \cite{golub10}.

For the second consideration, suppose that $\beta$ is a simple moment
with respect to the $b(x)$ given by Eq.~(\ref{monomial}). Then the
highest $n$ for which $\avg{a_n,b}_{P_0} \neq 0$ is $n = \mu$ by
Lemma~\ref{lem_poly}. Thus, a promising choice is
\begin{align}
S(x) &= a_\mu(x).
\label{a_mu}
\end{align}

\section{\label{sec_qbound}Purification bounds}
A closed-form solution for the Helstrom information of the mixed-state
submodel remains difficult to obtain. A standard technique in quantum
metrology is to bound it using a purified model and the monotonicity
of Helstrom information \cite{escher,demkowicz15,hayashi}.
\begin{lemma}
\label{lem_monotone}
Let $\{\omega_\theta:\theta \in \mathcal R \subseteq \mathbb R\}$ be a
regular model of states on $\mathcal H\otimes\mathcal H'$ and
$\{\tau_\theta = \trace'\omega_\theta:\theta \in \mathcal R\}$ be the
model on $\mathcal H$ generated by the partial trace. Then
$\{\tau_\theta\}$ is also regular, and
\begin{align}
\norm{\trace'_* S}_{\trace'\omega}^2 &\le \norm{S}_\omega^2,
\end{align}
where $S$ is the score of $\{\omega_\theta\}$ at the true $\omega$ and
$\trace'_* S$ is the score of $\{\tau_\theta\}$ at $\trace'\omega$.
\end{lemma}
The monotonicity is well established for finite-dimensional states
\cite{hayashi} and commonly assumed in the physics literature even for
infinite-dimensional states \cite{escher,demkowicz15}.  As I am unable
to find a rigorous proof for the general case in the literature, I
present one in Appendix~\ref{app_lem_monotone} for completeness.

For the submodel $\{\tau(P_\theta)\}$ in Sec.~\ref{sec_submodel}, the
natural purification given by Eq.~(\ref{natural_pur}) can be written
as
\begin{align}
\tau(P_\theta) &= \trace' \Psi_\theta,
\\
\ket{\Psi_\theta} &= 
e^{-i\hat k\otimes\hat x}\ket{\psi}\otimes \sqrt{g_\theta(\hat x)}
\ket{\phi_0},
\end{align}
where $\ket{\phi_0}$ is a purification of $P_0$. The score of
$\{\Psi_\theta:\theta \in (-c,c)\}$ is then $\hat S = S(\hat x)$, and
Lemma~\ref{lem_monotone} gives
\begin{align}
\norm{\tau_* S}_{\tau(P_0)}^2 &= 
\norm{\trace'_* \hat S}_{\trace'\Psi_0}^2 
\le \norm{\hat S}_{\Psi_0}^2 = \norm{S}_{P_0}^2 = 1,
\end{align}
which is a loose bound that does not depend on $\ket{\psi}$ of the
imaging system. I therefore turn to the alternative purification in
Lemma~\ref{lem_pur}.  Each $P_\theta$ given by Eq.~(\ref{P_theta})
satisfies the condition of infinite and bounded support for
Lemma~\ref{lem_pur} as long as the true $P_0$ satisfies it, as each
$P_\theta$ is dominated by $P_0$ and the $g_\theta(x)$ in
Eq.~(\ref{g}) is strictly positive. I can then use Lemma~\ref{lem_pur}
to write
\begin{align}
\tau(P_\theta) &= \trace'' \Phi_\theta,
\\
\ket{\Phi_\theta} &= \sum_{p=0}^\infty \sum_{n=0}^p \frac{(-i\hat k)^p}{p!}\chol_{pn}(\theta)
\ket{\psi}\otimes\ket{n}.
\label{Phi_theta}
\end{align}
where the parameter dependence comes from the Cholesky factor
$\chol(\theta)$ of the Hankel matrix
\begin{align}
  H_{pq}(\theta) &= \Avg{x^p,x^q}_{P_\theta} = 
                   \int  x^{p+q} P_\theta(dx).
\label{H_theta}
\end{align}
Note that the isometry used to generate this purification depends on
$P$ and therefore $\theta$ for the submodel, so the resulting Helstrom
information may differ from that using the original purification.

Let $S_\Phi$ be the score of $\{\Phi_\theta:\theta \in
(-c,c)\}$. Then, by Lemma~\ref{lem_monotone},
\begin{align}
\norm{\tau_* S}^2_{\tau(P_0)}  &=
\norm{\trace''_* S_\Phi}^2_{\trace'' \Phi_0} \le \norm{S_\Phi}^2_{\Phi_0}.
\label{monotone2}
\end{align}
Figure~\ref{purifications} illustrates the concepts introduced thus
far from the geometric perspective. While a closed-form solution for
$\norm{S_\Phi}^2_{\Phi_0}$ is still intractable, its scaling with the
true object size $\Delta$ defined by Eq.~(\ref{Delta}) in the
asymptotic limit $\Delta \to 0$ can be proved. Without loss of
generality, $\Delta$ can be measured in dimensionless ``Airy units''
relative to an effective bandwidth of the imaging system.
$\Delta \ll 1$ then defines the regime of subdiffraction objects.

\fig{0.48}{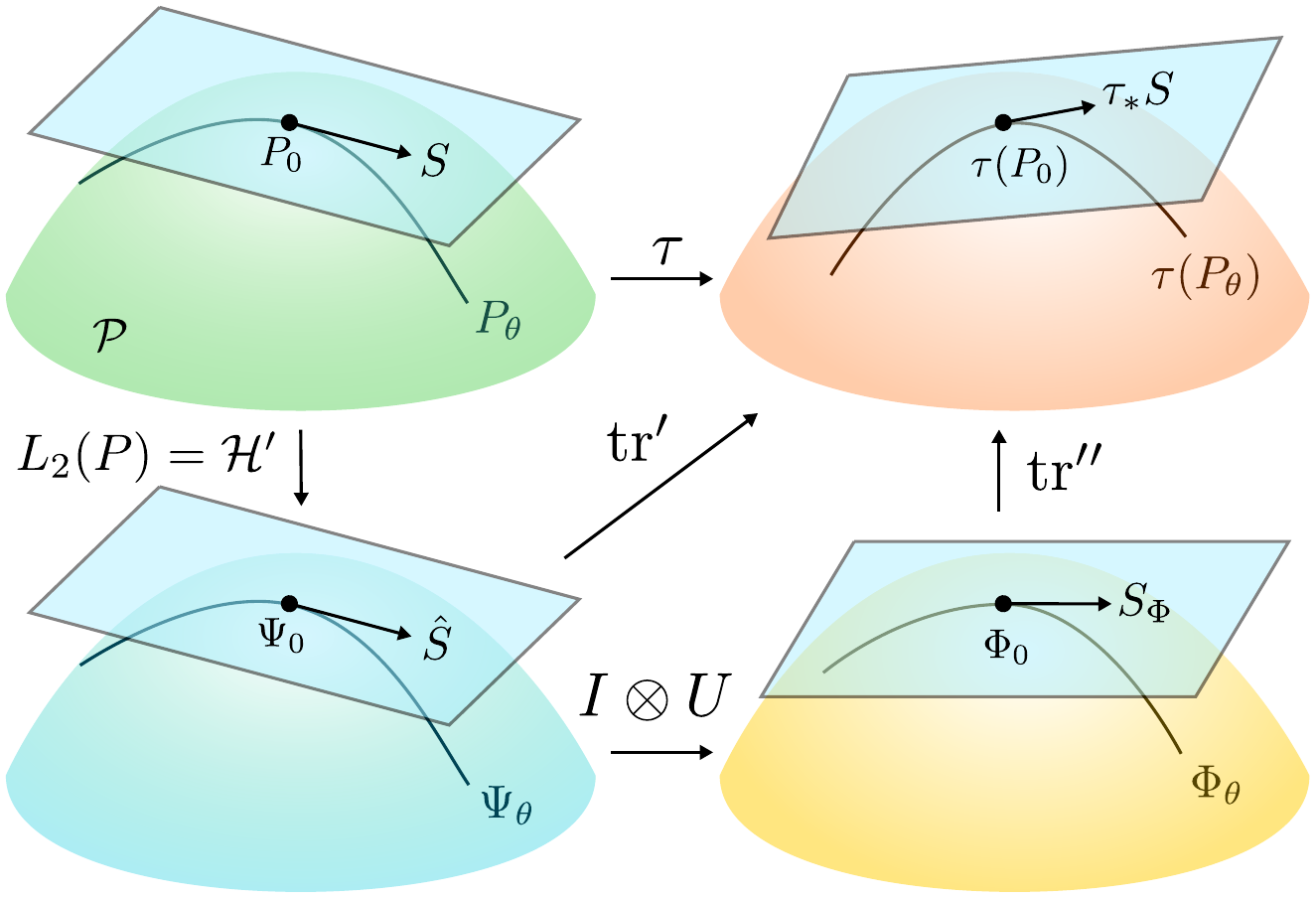}{A geometric picture that summarizes all the
  purifications, parametric submodels, scores, and maps that are
  involved in the proof of Theorem~\ref{thm_qbound}.  With the maps
  $\tau$, $\trace'$, and $\trace''$ all leading to the same model
  $\{\tau(P_\theta)\}$, the final score can be expressed as the
  pushforwards
  $\tau_* S = \trace'_* \hat S = \trace''_* S_\Phi$.}

In the following, the asymptotic notations $O[\el(\Delta)]$ (order at
most $\el(\Delta)$), $o[\el(\Delta)]$ (order smaller than
$\el(\Delta)$), $\Omega[\el(\Delta)]$ (order at least $\el(\Delta)$),
and $\Theta[\el(\Delta)]$ (order exactly $\el(\Delta)$)
will be used \cite{knuth76,miller06}.

\begin{theorem}
\label{thm_qbound}
Assume the semiparametric model given by
Eqs.~(\ref{rho})--(\ref{tau}), (\ref{G}), and (\ref{set}). Assume that
$\hat k$ is bounded. Assume further that $P_0$ is in the Szeg\H{o}
class \cite{widom66}, viz., $P_0$ has a density $f_0(x)$ with respect
to the Lebesgue measure on a compact interval $[c_1,c_2]$,
$c_1 < c_2$, such that
\begin{align}
\int_{c_1}^{c_2} \frac{\ln f_0(x)}{\sqrt{(c_2-x)(x-c_1)}} dx > -\infty.
\label{szego}
\end{align}
For the parameter of interest given by Eq.~(\ref{beta}) and
\begin{align}
b(x) &= x^\mu + o(\Delta^\mu),
\quad
\mu \in \mathbb N_1,
\label{b_gen}
\end{align}
the generalized Helstrom bound is
\begin{align}
\tilde{\mathsf H} &= \frac{\Omega(\Delta^{2\floor{\mu/2}})}{N}.
\label{qbound}
\end{align}
\end{theorem}
The proof is deferred to Appendix~\ref{app_thm_qbound}.

While Theorem~\ref{thm_qbound} appears to be identical to the main
result in Ref.~\cite{tsang19}, the rigor has been much improved
here. Most notably,
\begin{enumerate}
\item Lemma~\ref{lem_submodel} ensures that the 1D-submodel bound is
  valid for an infinite-dimensional model, avoiding the use of the
  questionable parametric model and matrix algebra in
Refs.~\cite{zhou19,tsang19}.

\item Lemma~\ref{lem_monotone} ensures that the monotonicity of the
  Helstrom information remains true for infinite-dimensional states.

\item The conditions for Theorem~\ref{thm_qbound} have been rigorously
  established, avoiding the questionable algebra of possibly unbounded
  operators in Ref.~\cite{tsang19}.
\end{enumerate}
The important physical implication of Theorem~\ref{thm_qbound} is that
the mean-square error $\mathsf E$ of any unbiased estimator must
decrease with $\Delta$ more slowly than the signal
\begin{align}
\beta^2 = O(\Delta^{2\mu}),
\end{align}
so the signal-to-noise ratio is
\begin{align}
\frac{\beta^2}{\mathsf E} &= N O(\Delta^{2\ceil{\mu/2}}),
\end{align}
and the moments of smaller subdiffraction objects are harder to
estimate, especially for higher orders.


Notice that Eq.~(\ref{b_gen}) generalizes the theorem for a larger
class of parameters, beyond the simple moments considered in previous
works \cite{tsang17,tsang18a,tsang19,tsang19b,zhou19}. I call a moment
associated with Eq.~(\ref{b_gen}) a generalized moment of order $\mu$.
The theorem shows that, regardless of the $o(\Delta^\mu)$ terms, the
quantum limits for generalized moments of the same order have the same
leading-order term.  This fact will be useful in the proof of the
optimality of SPADE in Sec.~\ref{sec_spade}.

The formalism here may also be able to deal with a $\beta(P)$ that is
nonlinear with respect to $P$, such as the entropy, in which case
$b(x)$ should be replaced by a gradient of $\beta(P)$ in the
$L_2(P_0)$ space, although this case is outside the scope of the
paper.

\section{\label{sec_spade}Spatial-mode demultiplexing with one or two
  modes}
Suppose that the spatial-frequency measure $Q$, defined in
Eq.~(\ref{Q}) with respect to $\hat k$ and $\ket{\psi}$ of the imaging
system, has an infinite and bounded support. Then the orthonormal
sequence
\begin{align}
\BK{\ket{\psi_n} \equiv (-i)^n \tilde a_n(\hat k)\ket{\psi}: n \in \mathbb N_0},
\end{align}
defined in terms of the orthonormal polynomials $\{\tilde a_n\}$ for
the measure $Q$ (as per Lemma~\ref{lem_poly}), can be used to
construct measurements of the light. The set is called the
point-spread-function-adapted (PAD) modes \cite{rehacek17,tsang18a},
generalizing the Hermite-Gaussian modes for a Gaussian $Q$.

It is interesting to note that the orthonormal polynomials $\{a_n\}$
for $P$ play central roles in previous sections, enabling the
purification given by Lemma~\ref{lem_pur} (as shown in
Appendix~\ref{app_lem_pur}) and also providing the score for the
submodel used in Theorem~\ref{thm_qbound}. Following similar steps, it
can be shown that the probability of each photon being projected into
a PAD mode is
\begin{align}
q_n(P) &\equiv \bra{\psi_n} \tau(P)\ket{\psi_n}
= \int  |C_n(x)|^2P(dx),
\\
C_n(x) &\equiv \bra{\psi_n} e^{-i\hat kx}\ket{\psi}
= \sum_{p=0}^\infty \frac{(-ix)^p i^n}{p!}\tilde\chol_{pn},
\label{C_tilde}
\end{align}
where $\tilde\chol_{pn}$ is now the Cholesky factor of the Hankel
matrix with respect to $Q$. A property of $\tilde\chol_{pn}$ is that
it is zero for $p < n$ and nonzero for $p = n$ (as per
Lemma~\ref{lem_poly}), so
\begin{align}
C_n(x) &= \frac{\tilde\chol_{nn}}{n!}x^n + O(\Delta^{n+1}),
\label{C_tilde2}
\\
q_n(P) &= r_n \beta_{2n}(P),
\quad
r_n \equiv \frac{\tilde\chol_{nn}^2}{n!^2},
\\
\beta_{2n}(P) &= \int \Bk{x^{2n} + O(\Delta^{2n+1})}P(dx).
\end{align}
In other words, each probability $q_n$ is proportional to a certain
generalized object moment $\beta_{2n}$ of even order $\mu = 2n$.

If $Q$ is Gaussian, say,
\begin{align}
Q(dk) &= \sqrt{\frac{2}{\pi}} e^{-2k^2} dk,
\end{align}
then its support is unbounded, but the preceding discussion still
holds, as Eq.~(\ref{C_tilde}) has the closed-form solution
\cite{tsang17}
\begin{align}
C_n(x) &= e^{-x^2/8} \frac{(x/2)^n}{\sqrt{n!}}.
\end{align}
There may exist more general conditions on $Q$ for the results in this
section to be valid, but such a generalization does not seem to be
interesting from the physics perspective and is therefore not pursued
in this work.

Let $\{\mathcal N_n: n \in \mathbb N_0\}$ be the integrated photon counts from
the PAD-mode projections over the $M$ temporal modes. The
expected value of each count is
\begin{align}
\expect\bk{\mathcal N_n} &= M\epsilon q_n = N r_n \beta_{2n}.
\label{Qn_mean}
\end{align}
Assuming a known $N$ (the theory for an unknown $N$ is similar
\cite{tsang18a,tsang19b}), an unbiased estimator of $\beta_{2n}(P)$ is
therefore
\begin{align}
\check \beta_{2n} &= \frac{\mathcal N_n}{r_n N}.
\label{estimator}
\end{align}

To estimate an odd generalized moment, consider the so-called
interferometric-PAD (iPAD) modes \cite{tsang17,tsang18a}
\begin{align}
\ket{\psi_n^+} &\equiv \frac{\ket{\psi_n}+\ket{\psi_{n+1}}}{\sqrt{2}},
&
\ket{\psi_n^-} &\equiv \frac{\ket{\psi_n}-\ket{\psi_{n+1}}}{\sqrt{2}}.
\end{align}
Let the integrated photon counts resulting from projections in this
pair of modes be $\mathcal N_n^+$ and $\mathcal N_n^-$. Proceeding in
the same way as before, it can be shown that
\begin{align}
\expect(\mathcal N_n^+) &= \frac{N}{2}
\Bk{r_n \beta_{2n} + s_n \beta_{2n+1} + O(\Delta^{2n+2})},
\label{Qn1_mean}
\\
\expect(\mathcal N_n^-) &= \frac{N}{2}
\Bk{r_n \beta_{2n} - s_n \beta_{2n+1} + O(\Delta^{2n+2})},
\label{Qn2_mean}
\end{align}
where $s_n \equiv 2\tilde\chol_{nn}\tilde\chol_{n+1\,n+1}/[n!(n+1)!]$ 
and
\begin{align}
\beta_{2n+1}(P) &= \int \Bk{x^{2n+1} + O(\Delta^{2n+2})}P(dx)
\end{align}
is a certain odd generalized moment.  An unbiased estimator of
$\beta_{2n+1}(P)$ is therefore
\begin{align}
\check\beta_{2n+1} &= \frac{\mathcal N_n^+-\mathcal N_n^-}{s_n N}.
\label{estimator_odd}
\end{align}
See Refs.~\cite{tsang17,tsang18a} for further details about how bases
may be constructed from the PAD and iPAD modes.

The following proposition summarizes the SPADE performance.

\begin{proposition}
\label{prop_spade}
Assume that the spatial-frequency measure $Q$ either has an infinite
and bounded support or is Gaussian. Then, projections into the PAD
modes enable unbiased estimation of a certain set of even generalized
moments $\{\beta_{2n}: n \in \mathbb N_0\}$, with the variance of each
estimator $\check\beta_{2n}$ given by
\begin{align}
\var\bk{\check\beta_{2n}} &= \frac{O(\Delta^{2n})}{N}\Bk{1 + O(\epsilon)},
\label{V2n}
\end{align}
while projections into the iPAD modes enable unbiased estimation of a
certain set of odd generalized moments
$\{\beta_{2n+1}: n \in \mathbb N_0\}$, with the variance of each
estimator $\check\beta_{2n+1}$ given by
\begin{align}
\var\bk{\check\beta_{2n+1}} &= \frac{O(\Delta^{2n})}{N} \Bk{1 + O(\epsilon)}.
\label{V2n1}
\end{align}
The variances achieve the quantum-limited $\Delta$ scalings given by
Eq.~(\ref{qbound}) in Theorem~\ref{thm_qbound}.
\end{proposition}
\begin{proof}
  Given Eq.~(\ref{rho}), $\mathcal N$ is a binomial process
  \cite{falk11}, which leads to
\begin{align}
\var(\mathcal N_n) &= N q_n(1-\epsilon q_n) = 
N r_n\beta_{2n} \Bk{1+O(\epsilon)}.
\end{align}
The variance of the estimator given by Eq.~(\ref{estimator})
is then
\begin{align}
\var\bk{\check\beta_{2n}} &= \frac{\var(\mathcal N_n)}{r_n^2 N^2} 
= \frac{\beta_{2n}}{r_n N}\Bk{1+O(\epsilon)},
\end{align}
which gives Eq.~(\ref{V2n}), since $\beta_{2n} = O(\Delta^{2n})$.  The
derivation of Eq.~(\ref{V2n1}) using the basic properties of a
binomial process is similar.
\end{proof}
Note that the conditions for Proposition~\ref{prop_spade} are more
relaxed than those for Theorem~\ref{thm_qbound}.

The important point is that, although a finite number of PAD or iPAD
modes cannot measure a simple moment of the object exactly, they can
provide exact unbiased estimators of the moments they are naturally
measuring. By considering the latter as the parameters of interest,
Proposition~\ref{prop_spade} proves the quantum optimality of SPADE
with just one or two modes, at least in terms of the $\Delta$
scalings.

\section{\label{sec_direct}A bound for direct imaging}
The semiclassical model of direct imaging is to assume a probability
density $\eta(\xi,P)$ for each photon with respect to the Lebesgue
measure that obeys
\begin{align}
\eta(\xi,P) &= \int h(\xi-x) P(dx),
\label{eta}
\end{align}
where $h$ is the nonnegative point-spread function for direct
incoherent imaging.  For a diffraction-limited system, it is given by
\begin{align}
h(\xi) &= \abs{\frac{1}{\sqrt{2\pi}}\intall e^{ikx}\braket{k|\psi} dk}^2,
\end{align}
where $\braket{k|\psi}$ is the optical transfer function
\cite{goodman_stat}. Let $\nu^{(M,\epsilon)}[P]$ denote the
probability measure for the photon counts over $M$ temporal
modes. Previous studies \cite{tsang17,tsang18a,tsang19b} have
suggested that, for the semiparametric model
$\{\nu^{(M,\epsilon)}[P]:P\in\mathcal P\}$ based on Eq.~(\ref{set}),
the Cram\'er-Rao bound for the estimation of a simple moment with
direct imaging is
\begin{align}
\mathsf E &\ge \tilde{\mathsf C} = \frac{\Theta(1)}{N}.
\label{CRB_direct}
\end{align}
Unfortunately, this result is rigorously proven only when the location
family $\{h(\xi-x): x \in \supp P_0\}$ satisfies a special statistical
property called completeness \cite{lehmann_romano}; see
Ref.~\cite[Sec.~6.5]{bickel93}, Ref.~\cite[Sec.~25.5.2]{vaart}, and
Ref.~\cite{tsang19b}. While a Gaussian $h$ satisfies the property,
completeness turns out to be hard to prove more generally, and
Eq.~(\ref{CRB_direct}) remains questionable for any non-Gaussian $h$.

I turn again to the parametric-submodel approach in
Sec.~\ref{sec_submodel}, which leads
to the following proposition.

\begin{proposition}
\label{prop_direct}
Assume that the derivatives $h^{(n)}(\xi)$ of the point-spread
function $h(\xi)$ for direct imaging exist, up to order $n = \mu+1$.
If there exist nonnegative $x$-independent functions
$\underline{h}(\xi)$ and $\bar{h}(\xi)$ such that, for all
$\xi \in \mathbb R$ and for all $x \in \supp P_0$,
\begin{align}
h(\xi-x) &\ge \underline{h}(\xi),
\label{h0}
\\
|h^{(\mu+1)}(\xi-x)| &\le \bar{h}(\xi),
\label{h_mu1}
\end{align}
and the functions satisfy
\begin{align}
\intall \frac{[h^{(\mu)}(\xi)]^2}{\underline{h}(\xi)} d\xi &< \infty,
&
\intall \frac{[\bar{h}(\xi)]^2}{\underline{h}(\xi)} d\xi &< \infty,
\label{finite_int}
\end{align}
then the Cram\'er-Rao bound for the estimation of a generalized moment
of order $\mu$ is
\begin{align}
\tilde{\mathsf C} &= \frac{\Omega(1)}{N}.
\label{CRB_direct_bound}
\end{align}
\end{proposition}

The proof is deferred to Appendix~\ref{app_prop_direct}.

The conditions for Proposition~\ref{prop_direct} seem specific, but
they can be checked for simple functions, such as
\begin{align}
h(\xi) &= \frac{d_1}{(\xi/d_2)^{2p} + 1},
&
h(\xi) &= d_1 e^{-(\xi/d_2)^{2p}},
\label{simple_h}
\end{align}
for some positive constants $d_1,d_2$ and positive integer $p$.  The
tails of $h^{(p)}(\xi-x)$ for these functions decay at least as fast
as the tails of $h(\xi)$, while an $\underline{h}(\xi)$ to lower-bound
$h(\xi-x)$ can be obtained by replacing $\xi^2$ by
$(|\xi|+\Delta_0)^2$ in Eqs.~(\ref{simple_h}) for some $\Delta_0$ that
upper-bounds all $\Delta$ of interest, as that would give
$|\xi-x| \le |\xi| + |x| \le |\xi|+\Delta \le |\xi|+\Delta_0$ and a
lower bound on $h(\xi-x)$ for all $x \in \supp P_0$. Then
$\underline{h}(\xi)$ remains positive and has tails that decay in the
same way as $h(\xi)$, leading to the finite integrals in
Eqs.~(\ref{finite_int}). The advantage of
Proposition~\ref{prop_direct} over the previous approaches in
Refs.~\cite{tsang17,tsang18a,tsang19b} is that the conditions for the
former are still easier to check than proving the completeness
property, and Proposition~\ref{prop_direct} can remain valid even for
an incomplete location family.

The physical implication of Proposition~\ref{prop_direct} is that, for
a generalized moment of order $\mu \ge 2$, both the quantum limit
given by Theorem~\ref{thm_qbound} and the SPADE performance given by
Proposition~\ref{prop_spade} are substantial improvements over direct
imaging. That said, there exist counterexamples in which the
point-spread function has zeros and the conditions for
Proposition~\ref{prop_direct} are violated. The argument by Pa\'ur and
coworkers \cite{paur18,paur19}, in particular, suggests that the zeros
can enhance the Fisher-information integral for the submodel by a
$\Theta(1/\Delta)$ factor, leading to
\begin{align}
\tilde{\mathsf C} & = \frac{\Omega(\Delta)}{N},
\label{conjecture}
\end{align}
but despite some numerical evidence (unpublished), I am unable to
prove this bound in a general fashion. It therefore remains an open
problem whether Eq.~(\ref{conjecture}) is the ultimate limit to direct
imaging, or by how much the zeros in its point-spread function can
improve it.

\section{\label{conclusion}Conclusion}
The key results of this work are the quantum limit given by
Theorem~\ref{thm_qbound}, the SPADE performance given by
Proposition~\ref{prop_spade}, and the direct-imaging bound given by
Proposition~\ref{prop_direct}. They confirm rigorously the prior
intuition that the moments of smaller subdiffraction objects are
harder to estimate, especially for higher orders, although SPADE can
estimate them with optimal error scalings, while direct imaging is
unlikely to be nearly as efficient. Beyond the improved rigor, a
useful advance is the generalization of the results for a larger class
of moments, such that the quantum optimality of SPADE with a finite
number of modes is proved and its experimental demonstration becomes
much easier.

The price to pay for the generality of the results here is their
imprecise nature in terms of the asymptotic notions. More precise
results can be obtained numerically for more special cases, as has
been done in Refs.~\cite{tsang17,tsang18a,tsang19b} regarding the
SPADE and direct-imaging performances. To compute concrete quantum and
classical limits, the parametric-submodel approach should help, as it
is able to give bounds for an infinite-dimensional model through 1D
submodels. The submodel bounds can be computed numerically without
resorting to purifications, at least for special cases of the true
measure $P_0$.

Other interesting future directions include more rigorous proofs of
the convergence of thermal models to the binomial and Poisson models
considered here, the study of more general types of parameters,
Bayesian and minimax approaches \cite{gill95,tsang18,tsang20b}, and,
of course, experimental demonstrations of quantum-limited moment
estimation. Beyond imaging, the model and the results here may be
applied or generalized to other sensing applications, such as the
estimation of diffusion parameters in phase estimation and
optomechanics \cite{vidrighin,ng16}. The general principles
established in this work thus give rigorous underpinnings to the
foundations of imaging theory and possibly beyond.

\section*{Acknowledgments}
This work is supported by the National Research Foundation (NRF)
Singapore, under its Quantum Engineering Programme (Grant No.~QEP-P7).

\appendix

\section{\label{app_lem_pur}Proof of Lemma~\ref{lem_pur}}
Before proving Lemma~\ref{lem_pur}, I collect some basic facts in
the following lemma.
\begin{alemma}
\label{lem_poly}
In this lemma, the sample space is $\mathbb R$, the inner product is
given by Eq.~(\ref{inner_P}), and the measure $P$ has finite moments
and $\#\supp P = J$, where $J$ may be infinite. Let $A_{(j)}$ denote
the $(j+1)\times (j+1)$ upper-left submatrix of a matrix $A$ for a
$j \in \mathbb N_{\bko{0,J}} \equiv \{j \in \mathbb N_0: j < J\}$.
\begin{enumerate}
\item $H_{(j)}$ for the Hankel matrix given by Eq.~(\ref{hankel}) is
  positive-definite ($H_{(j)}> 0$) for any $j \in \mathbb N_{\bko{0,J}}$.

\item The monomials $\{x^p: p \in \mathbb N_{\bko{0,J}}\}$ are
  linearly independent.

\item Define the orthonormal polynomials
$\{a_n(x): n \in \mathbb N_{\bko{0,J}}\}$ as
\begin{align}
a_n(x) &= \sum_{p=0}^n A_{np} x^p,
\end{align}
where $A$ is the lower-triangular matrix ($A_{np} = 0$ if $p > n$)
obtained by applying the Gram-Schmidt procedure to the monomials, such
that
\begin{align}
\Avg{a_n,a_m}_P &= \delta_{nm}.
\end{align}
One can make $A_{nn} > 0$ for all $n$. The orthonormality implies, for
any $j \in \mathbb N_{\bko{0,J}}$,
\begin{align}
A_{(j)} H_{(j)} A_{(j)}^{\top} &= I_{(j)},
&
H_{(j)}^{-1} &= A_{(j)}^{\top}A_{(j)},
\end{align}
where $\top$ denotes the transpose and $I$ is the identity matrix.

\item Let $\chol$ be the lower-triangular Cholesky factor of $H$.
  Then, for any $j,n,p \in \mathbb N_{\bko{0,J}}$,
\begin{align}
\chol_{(j)} \chol_{(j)}^{\top} &= H_{(j)},
\quad
\chol_{(j)} = A_{(j)}^{-1},
\label{LHA}
\\
\chol_{nn} &= \frac{1}{A_{nn}} > 0,
\\
x^p &= \sum_{n=0}^p \chol_{pn} a_n(x),
\quad
\Avg{x^p,a_n}_P = \chol_{pn}.
\end{align}
\end{enumerate}
\end{alemma}

\begin{proof}[Proof of Lemma~\ref{lem_poly}]
  I prove only statement 1 here; the rest is basic linear algebra;
  see, for example, Ref.~\cite{golub10}. Write $C = \supp P$ for
  brevity in this proof.  Recall that the support is defined uniquely
  by the following conditions \cite[Theorem~2.1]{parthasarathy67}:
\begin{enumerate}
\item $P(C) = 1$.
\item If a closed set $D$ satisfies $P(D) = 1$, then $C \subseteq D$.
\end{enumerate}
In other words, the support is the smallest closed set with unit
probability. The second condition is equivalent to the condition that
any closed strict subset $D$ of $C$ must give $P(D) < 1$.

Given any column vector $\el \in \mathbb R^{j+1}$ with
$\el^\top \el \neq 0$, let
\begin{align}
\tilde\el(x) &\equiv \sum_{p=0}^j \el_p x^p,
&
Z &\equiv \BK{x \in C: \tilde\el(x) = 0}.
\end{align}
Then
\begin{align}
\Bk{\tilde\el(x)}^2 &> 0 \quad \forall x \in C - Z.
\label{pj_pos}
\end{align}
As any nonzero polynomial of degree $j$ has at most $j$ isolated
zeros, $\# Z \le j$ and $Z$ is closed. Since $\# C =J > j \ge \#Z$,
$C - Z$ cannot be empty, $Z$ is a strict subset of $C$, and by the
definition of $C$,
\begin{align}
P(Z) &< 1,
&
P(C-Z) &= P(C)-P(Z) > 0.
\end{align}
Now consider
\begin{align}
\el^\top H_{(j)} \el &= \int_{C}[\tilde\el(x)]^2 P(dx)
= \int_{C-Z} [\tilde\el(x)]^2 P(dx).
\label{step_pd}
\end{align}
I wish to prove that the last expression is not zero by contradiction.
Suppose that it is zero. Then, by Proposition~4.1.7 in
Ref.~\cite{parthasarathy05},
\begin{align}
[\tilde\el(x)]^2  &= 0 \quad \textrm{a.e.\ on }C - Z,
\label{pj_0}
\end{align}
where a.e.\ denotes almost everywhere. Equation~(\ref{pj_0}) implies
the existence of a set $Y$ with $P(Y) = 0$ such that $[\tilde\el(x)]^2 = 0$
for all $x \in C-Z-Y$. $C-Z-Y$ here is not empty, because
$P(C-Z-Y) = P(C-Z) > 0$. Then, for any $x \in C-Z-Y \subseteq C-Z$,
the $[\tilde\el(x)]^2 = 0$ statement contradicts Eq.~(\ref{pj_pos}). By
contradiction, Eq.~(\ref{step_pd}) cannot be
zero and must be strictly positive.  As $\el$ is arbitrary,
$H_{(j)} > 0$.
\end{proof}

\begin{proof}[Proof of Lemma~\ref{lem_pur}]
  If the support of $P$ is infinite and bounded, the set
  $\{a_n(x):n\in \mathbb N_0\}$ defined in Lemma~\ref{lem_poly} is an
  orthonormal basis of $L_2(P)$
  \cite[Proposition~6.4.1]{parthasarathy05}. Let its $\mathcal H'$
  representation be
\begin{align}
a_n(x) &\leftrightarrow a_n(\hat x)\ket{\phi} \equiv \ket{\varphi_n}.
\end{align}
Then $\{\ket{\varphi_n}:n\in\mathbb N_0\}$ is the same orthonormal
basis in the new representation, and a resolution of the identity
$I_{\mathcal H'}$ in the strong operator topology
\cite[Sec.~1.6]{davies} is
\begin{align}
I_{\mathcal H'} &= \sum_{n=0}^\infty \ket{\varphi_n}\bra{\varphi_n}.
\end{align}
An isometry $U:\mathcal H' \to \mathcal H''$ can then be written as
\begin{align}
U &= \sum_{n=0}^\infty \ket{n}\bra{\varphi_n},
\label{iso}
\end{align}
also in the strong operator topology. A proof is as follows.
Let 
\begin{align}
I_j &\equiv \sum_{n=0}^j \ket{\varphi_n}\bra{\varphi_n},
&
U_j &\equiv \sum_{n=0}^j \ket{n}\bra{\varphi_n}.
\end{align}
Given any $\ket{\el} \in \mathcal H'$ and $J > j$, consider
\begin{align}
\norm{U_J \ket{\el} - U_j\ket{\el}}^2  
&= \sum_{n=j+1}^J \abs{\braket{\varphi_n|\el}}^2 = s_J - s_j,
\label{strong_U}
\end{align}
where $\norm{\ket{\cdot}} \equiv \sqrt{\braket{\cdot|\cdot}}$ and
$s_j \equiv \sum_{n=0}^j |\braket{\varphi_n|\el}|^2$.  $\{s_j\}$
is monotonic, and since $\{\ket{\varphi_n}\}$ is an orthonormal
sequence, $\{s_j\}$ is bounded by Bessel's inequality
\cite[Eq.~(3.24)]{debnath05} and therefore convergent
\cite[Theorem~3.14]{rudin_baby} and Cauchy
\cite[Theorem~3.11]{rudin_baby}. With the Cauchy property, given any
$\varepsilon > 0$, there exists a $j_0$ such that, for all
$J \ge j \ge j_0$, $s_J - s_j < \varepsilon^2$ and
$\norm{U_J \ket{\el} - U_j\ket{\el}} = \sqrt{s_J - s_j} <
\varepsilon$.  Thus, $\{U_j\ket{\el}\}$ is also Cauchy.  Since a
Hilbert space is by definition complete
\cite[Definition~3.3.1]{debnath05}, $\{U_j\ket{\el}\}$ converges
\cite[Definition~1.4.5]{debnath05}, and since $\ket{\el}$ is
arbitrary, $\{U_j\}$ converges strongly. By the same argument,
$\{I_j\}$ and $\{U_j^\dagger\}$ also converge strongly.  The strong
convergences imply that $\{U_j^\dagger U_j = I_j\}$ converges strongly
to $U^\dagger U = I_{\mathcal H'}$.

If the support of $P$ is bounded, $\hat x$ can be assumed to be a
bounded operator with a finite operator norm $\opnorm{\hat x}$, since
$\norm{\hat x\ket{\el}} = \norm{x \el}_P\le \Delta \norm{\el}_P =
\Delta \norm{\ket{\el}}$ for any $\el \leftrightarrow \ket{\el}$.  As
$\hat k$ is also assumed to be bounded, the exponential operator 
in Eq.~(\ref{natural_pur}) can be expressed as
\begin{align}
e^{-i\hat k\otimes\hat x} = 
\sum_{p=0}^\infty \frac{(-i\hat k)^p\otimes\hat x^p}{p!}
\end{align}
in the sense of operator-norm convergence
\cite[Sec.~VIII.4]{reed_simon}, which implies strong convergence. Let
\begin{align}
  F_j &\equiv \sum_{p=0}^j \frac{(-i\hat k)^p\otimes\hat x^p}{p!}.
\end{align}
As $\{U_j\}$ and $\{F_j\}$ converge strongly,
$\{(I_{\mathcal H}\otimes U_j) F_j\}$ also converges strongly.  With
the aid of Lemma~\ref{lem_poly}, I obtain
\begin{align}
\hat x^p \ket\phi &= \sum_{n=0}^p \chol_{pn} a_n(\hat x)\ket\phi
= \sum_{n=0}^p \chol_{pn} \ket{\varphi_n},
\label{xp_phi}
\\
\ket{\Phi^j} &\equiv \bk{I_{\mathcal H}\otimes U_j} 
F_j \ket{\psi}\otimes\ket{\phi}
\\
&= \sum_{p=0}^j \sum_{n=0}^p \frac{(-i\hat k)^p}{p!}\chol_{pn}\ket{\psi}\otimes\ket{n}.
\label{PhiN}
\end{align}
The strong convergence of $\{(I_{\mathcal H}\otimes U_j) F_j\}$
implies the strong convergence of $\{\ket{\Phi^j}\}$, leading to
Eq.~(\ref{Phi}).
\end{proof}

\begin{remark}
  Throughout this paper, the convergence of a sequence of
  Hilbert-space elements is always assumed to be strong.
\end{remark}

\section{\label{app_lem_submodel}Proof of Lemma~\ref{lem_submodel}}
\begin{proof}[Proof of Lemma~\ref{lem_submodel}]
  Equation~(\ref{H_formula}) comes from the fact that
  $\{e_1 \equiv \el/\norm{\el}_\rho\}$ is the orthonormal basis of
  $\spn \el$ and
\begin{align}
\Pi(\delta|\spn \el) &= \avg{\delta,e_1}_\rho e_1 
= \frac{\avg{\delta,\el}_\rho}{\norm{\el}_\rho^2} \el.
\end{align}
Of course, $\Pi(\delta|\spn 0) = 0$.

Now consider the supremum in Eq.~(\ref{supremum}), which is defined by
two conditions \cite{rudin_baby}:
\begin{enumerate}
\item $\tilde{\mathsf H}$ is an upper bound on the set
  $\{\mathsf H(\el): \el \in \spn\{S\}\}$.
\item $\tilde{\mathsf H} - \varepsilon$ is not an upper bound for any
  $\varepsilon > 0$.
\end{enumerate}
Since $\deff \in \mathcal D$ \cite{tsang20},
  $\mathsf H(\el)$ can also be expressed as
\begin{align}
\mathsf H(\el) &= \norm{\Pi(\deff|\spn \el)}_\rho^2.
\end{align}
Then
\begin{align}
\tilde{\mathsf H} &= \norm{\deff}_\rho^2 \ge 
\norm{\Pi(\deff|\spn \el)}_\rho^2 = \mathsf H(\el)
\quad
\forall \el,
\end{align}
by the shrinking property of the projection \cite{debnath05}. The
first condition is therefore satisfied.

Assume now $\deff \neq 0$, for
$\tilde{\mathsf H} = \mathsf H(\el) = 0\ \forall \el$ otherwise and
the lemma is trivial. As $\deff \in \mathcal T$ and
$\mathcal T$ is the closure of $\spn\{S\}$, there exists a sequence
$\{\el_j \}\subseteq \spn\{S\}$ that converges to $\deff$
\cite[Theorem~1.3.23]{debnath05}.  The convergence means that
$\norm{\el_j}_\rho \to \norm{\deff}_\rho$ and
$\avg{\el,\el_j}_\rho \to \avg{\el,\deff}_\rho$ for any
$\el \in \mathcal T$ \cite[Theorem~3.3.12]{debnath05}. With
$\norm{\deff}_\rho > 0$, I can assume
$\norm{\el_j}_\rho > 0$ for all $j\ge J$ and a sufficiently large
$J$. Thus,
\begin{align}
\mathsf H(\el_j) &= 
\frac{\avg{\deff,\el_j}_\rho^2}{\norm{\el_j}_\rho^2}
\to \norm{\deff}_\rho^2 = \tilde{\mathsf H}.
\end{align}
This limit implies that, given any $\varepsilon > 0$, there exists a
$\el \in \{\el_j \}\subseteq \spn\{S\}$ such that
\begin{align}
|\tilde{\mathsf H}-\mathsf H(\el)| &= 
\tilde{\mathsf H}-\mathsf H(\el) < \varepsilon,
&
\mathsf H(\el) &> \tilde{\mathsf H} -\varepsilon.
\end{align}
Thus, the second condition is also satisfied.
\end{proof}

\section{\label{app_lem_monotone}Proof of Lemma~\ref{lem_monotone}}
Before proving Lemma~\ref{lem_monotone}, I present a couple of needed
lemmas.  The first is a trivial result regarding the partial trace and
the $L_2$ spaces related by it. I write $I_{\mathcal H'} = I$ for
brevity in this appendix.
\begin{alemma}
\label{lem_partial}
Let $\omega$ be a state on $\mathcal H\otimes\mathcal H'$ and
$\tau = \trace'\omega$.  Then
\begin{align}
\Avg{\el,\elalt}_\tau &= \Avg{\el\otimes I,\elalt\otimes I}_\omega
\quad
\forall \el,\elalt \in L_2(\tau),
\label{inner_tau_omega}
\\
\norm{\el}_\tau &= \norm{\el\otimes I}_\omega
\quad
\forall \el \in L_2(\tau).
\label{norm_partial}
\end{align}
\end{alemma}
\begin{proof}[Proof of Lemma~\ref{lem_partial}] By the definition of
  partial trace \cite[Proposition~16.6]{parthasarathy92},
\begin{align}
\trace \el \tau &= \trace \bk{\el\otimes I}\omega 
\quad
\forall \el \in \mathcal B(\mathcal H).
\label{partial_trace}
\end{align}
Then, for all
$\el,\elalt \in \mathcal B(\mathcal H)$,
\begin{align}
\Avg{\el,\elalt}_\tau &= \trace \bk{\el\circ \elalt} \tau= 
\trace \Bk{\bk{\el\otimes I}\circ \bk{\elalt\otimes I}}\omega
\nonumber\\
&= \Avg{\el\otimes I,\elalt\otimes I}_\omega.
\label{step_bounded}
\end{align}
If any of $\el,\elalt \in L_2(\tau)$ are unbounded operators, let
$\{\el_{j}\}$ and $\{\elalt_{j}\}$ be sequences in the dense subset
$\mathcal B(\mathcal H)$ that converge to them. Then
\cite[Theorem~3.3.12]{debnath05}
\begin{align}
\Avg{\el_j,\elalt_j}_\tau &\to \Avg{\el,\elalt}_\tau,
\label{hg}
\\
\Avg{\el_j\otimes I,\elalt_j\otimes I}_\omega
&\to 
\Avg{\el\otimes I,\elalt\otimes I}_\omega.
\label{hIgI}
\end{align}
As Eq.~(\ref{step_bounded}) holds for all
$\el_j,\elalt_j \in \mathcal B(\mathcal H)$, the right-hand sides of
Eqs.~(\ref{hg}) and (\ref{hIgI}) are also equal, proving
Eq.~(\ref{inner_tau_omega}). Eq.~(\ref{norm_partial}) is a direct
consequence of Eq.~(\ref{inner_tau_omega}).
\end{proof}
I now present a useful lemma regarding a quantum generalization of the
conditional expectation, which appears in many other contexts
\cite{ohki18,tsang19c,hayashi}.  The proof follows that of Theorem~6.1
in Ref.~\cite{hayashi} for the finite-dimensional case.

\begin{alemma}
\label{lem_cond}
Let $\omega$ be a state on $\mathcal H\otimes\mathcal H'$ and
$\tau = \trace'\omega$. For each $X \in L_2(\omega)$,
$\avg{\el\otimes I,X}_\omega$ is a bounded linear
functional of $\el \in L_2(\tau)$. Define $\pi(X) \in L_2(\tau)$ as the
unique Riesz representation that obeys
\begin{align}
\Avg{\el,\pi(X)}_\tau &= \Avg{\el\otimes I,X}_\omega
\quad
\forall \el \in L_2(\tau).
\label{riesz}
\end{align}
Then the linear map $\Pi:L_2(\omega)\to L_2(\omega)$ defined as
\begin{align}
\Pi X &\equiv \pi(X)\otimes I
\label{Pi}
\end{align}
is a projection. In particular, $\pi$ satisfies the shrinking property
\begin{align}
\norm{\pi(X)}_\tau =
\norm{\Pi X}_\omega &\le \norm{X}_\omega \quad
\forall X \in L_2(\omega).
\label{contract}
\end{align}
\end{alemma}
\begin{proof}[Proof of Lemma~\ref{lem_cond}]
  First note that $\avg{\el\otimes I,X}_\omega$ is bilinear with
  respect to $\el \in L_2(\tau)$ and $X \in L_2(\omega)$. For each
  $X \in L_2(\omega)$,
\begin{align}
\abs{\avg{\el\otimes I,X}_\omega} &\le  \norm{\el\otimes I}_\omega \norm{X}_\omega 
= \norm{\el}_\tau \norm{X}_\omega 
\label{bounded_functional}
\end{align} 
by the Cauchy-Schwarz inequality (CSI) and Lemma~\ref{lem_partial}, so
$\avg{\el\otimes I,X}_\omega$ is a bounded linear functional of $\el$,
and the Riesz representation theorem applies. Moreover, the theorem
gives
\begin{align}
\norm{\pi(X)}_\tau = 
\opnorm{\pi} = \sup_{\norm{\el}_\tau = 1} \abs{\avg{\el\otimes I,X}_\omega}
\le \norm{X}_\omega,
\label{bounded_map}
\end{align}
so $\pi:L_2(\omega)\to L_2(\tau)$ is a bounded linear map. 

Now consider the $\Pi$ map given by Eq.~(\ref{Pi}). It is linear,
because $\pi$ is linear. It is bounded, because, by
Lemma~\ref{lem_partial} and Eq.~(\ref{bounded_map}),
\begin{align}
\norm{\Pi X}_\omega &= \norm{\pi(X)}_\tau \le \norm{X}_\omega.
\end{align}
It is self-adjoint, because, for any
$X,Y \in L_2(\omega)$,
\begin{align}
\Avg{Y,\Pi X}_\omega &= \Avg{Y,\pi(X)\otimes I}_\omega 
= \Avg{\pi(Y),\pi(X)}_\tau
\\
&= \Avg{\pi(Y)\otimes I,X}_\omega
= \Avg{\Pi Y,X}_\omega.
\end{align}
It is also idempotent, because, for any $X,Y \in L_2(\omega)$,
\begin{align}
\Avg{Y,\Pi^2 X}_\omega &= \Avg{\Pi Y,\Pi X}_\omega = 
\Avg{\pi(Y),\pi(X)}_\tau \\
&= \Avg{Y,\Pi X}_\omega,
\end{align}
meaning that $\Pi^2 = \Pi$.  $\Pi$ is hence a projection
\cite[Theorem~4.7.7]{debnath05}, and Eq.~(\ref{contract}) is a basic
property.
\end{proof}

\begin{proof}[Proof of Lemma~\ref{lem_monotone}] Let
  $\mathfrak T(\mathcal H)$ be the Banach space of all trace-class
  self-adjoint operators on $\mathcal H$ \cite{holevo11}. The partial
  trace
  $\trace':\mathfrak T(\mathcal H\otimes\mathcal H') \to \mathfrak
  T(\mathcal H)$ is a bounded linear map \cite[p.~150]{davies} and
  therefore continuous \cite[Theorem~1.5.7]{debnath05}.  It follows
  that $\trace'$ is differentiable and the derivative (as a linear
  map) is given by $\trace'$ itself \cite[(8.1.3)]{dieudonne}. As
  $\{\omega_\theta\}$ is assumed to be regular, $\omega_\theta$ is
  differentiable at the truth. Thus,
  $\tau_\theta = \trace'\omega_\theta$ is also differentiable at the
  truth \cite[(8.2.1)]{dieudonne}, and
  $\dot\tau \in \mathfrak T(\mathcal H)$ is determined by the chain
  rule
\begin{align}
    \dot\tau &= \trace'\dot\omega.
\end{align}
Since $\{\omega_\theta\}$ is assumed to be regular, its score
$S \in L_2(\omega)$ exists. The definition of the partial trace as per
Eq.~(\ref{partial_trace}) and the definition of the score as per
Eq.~(\ref{score}) can then be used to give
\begin{align}
\trace \el \dot\tau &= \trace \bk{\el\otimes I} \dot\omega
= 
\Avg{\el\otimes I,S}_\omega
\quad
\forall \el \in \mathcal B(\mathcal H).
\label{pushforward_S}
\end{align}
The linear functional
$\trace \el \dot\tau = \avg{\el\otimes I,S}_\omega$ of $\el$ is
bounded with respect to $\norm{\el}_\tau$ by
Eq.~(\ref{bounded_functional}) and can therefore be extended uniquely
to be defined on $L_2(\tau)$. The score $\trace_*'S \in L_2(\tau)$ of
$\{\tau_\theta\}$ then exists by definition, $\{\tau_\theta\}$ is
regular, and 
\begin{align}
\Avg{\el,\trace_*'S}_\tau = \Avg{\el\otimes I,S}_\omega\quad
\forall \el \in L_2(\tau).
\end{align}
By Lemma~\ref{lem_cond}, the score must be equal to the conditional
expectation
\begin{align}
\trace'_* S &= \pi(S),
\end{align}
which observes the shrinking property
\begin{align}
\norm{\trace'_* S}_\tau 
&= \norm{\pi(S)}_\tau \le \norm{S}_\omega.
\end{align}
\end{proof}
\begin{remark}
  Throughout this paper, differentiability is always assumed to be in
  the Fr\'echet sense \cite{flett}.
\end{remark}

\section{\label{app_thm_qbound}Proof of Theorem~\ref{thm_qbound}}
Before proving Theorem~\ref{thm_qbound}, I need a few lemmas. The
first lemma gives $\supp P_\theta = \supp P_0$ under general conditions.

\begin{alemma}
\label{lem_support}
$g_\theta(x) > 0$ a.e.~$P_0$ on $\supp P_0$ implies that the
$P_\theta$ defined by Eq.~(\ref{P_theta}) and $P_0$ have the same
support.
\end{alemma}
\begin{proof}[Proof of Lemma~\ref{lem_support}]
  Write $C = \supp P_0$ for brevity in this proof.  First notice that,
  since $P_0(C) = 1$,
\begin{align}
P_\theta(C) &= \int_C g_\theta(x) P_0(dx) = 
\int g_\theta(x) P_0(dx) = 1,
\end{align}
and the first condition for $\supp P_\theta = C$ is satisfied. To
prove the second condition, consider, for any closed strict subset $D$
of $C$,
\begin{align}
P_\theta(D) &= \int_D g_\theta(x) P_0(dx) = 1 - \int_{C-D} g_\theta(x) P_0(dx).
\label{PD}
\end{align}
By the definition of $C$, $P_0(D) < 1$ and $P_0(C-D) > 0$.  I can then
prove that the last integral in Eq.~(\ref{PD}) is not zero by
contradiction. Suppose that it is zero. Then, by Proposition~4.1.7 in
Ref.~\cite{parthasarathy05},
\begin{align}
g_\theta(x) = 0 \quad \textrm{a.e.~$P_0$ on $C-D$}.
\label{g0ae}
\end{align}
This statement contradicts the assumption that $g_\theta(x) > 0$
a.e.~$P_0$ on $C$ by the following argument. The assumption implies
the existence of a set $Y_1$ with $P_0(Y_1) = 0$ such that
\begin{align}
g_\theta(x) > 0 \quad \forall x \in C - Y_1.
\label{g_pos}
\end{align}
On the other hand, Eq.~(\ref{g0ae}) implies the existence of a $Y_2$
with $P_0(Y_2) = 0$ such that
\begin{align}
g_\theta(x) = 0 \quad \forall x \in C-D-Y_2.
\label{g0}
\end{align}
The common domain $(C - Y_1)\cap (C-D-Y_2) = C- D - (Y_1 \cup Y_2)$ is
not empty because $P_0[C-D-(Y_1\cup Y_2)] = P_0(C-D) > 0$.  Thus, for
any $x$ in the common domain, the two statements in Eqs.~(\ref{g_pos})
and (\ref{g0}) contradict each other. By contradiction, the last
integral in Eq.~(\ref{PD}) must be strictly positive,
$P_\theta(D) < 1$, and the second condition for $\supp P_\theta = C$
is also satisfied.

\end{proof}
The second lemma provides some bounds on $g_\theta(x)$ given by
Eqs.~(\ref{g}) and (\ref{a_mu}), which lead to useful properties for
the submodels in Secs.~\ref{sec_submodel} and \ref{sec_qbound}.

\begin{alemma}
\label{lem_g}
For the $g_\theta(x)$ given by Eqs.~(\ref{g}) and (\ref{a_mu}), there
exist finite positive constants $c_3, c_4, c_5$ such that, for all
$|\theta| \le c < \infty$ and $|x| \le \Delta < \infty$,
\begin{align}
0 < c_3 \le g_\theta(x) &\le c_4 < \infty,
\label{c3}
\\
|\partial g_\theta(x)| &\le c_5 < \infty,
\label{c5}
\end{align}
where $\partial$ denotes $\partial/\partial\theta$. Moreover, defining
\begin{align}
\bar \el(\theta) &\equiv \int \el(x) g_\theta(x) P_0(dx)
\end{align}
for a $P_0$ that obeys Eq.~(\ref{Delta}) and a $P_0$-integrable
function $\el(x)$, $\partial$ can be taken inside the integral, viz.,
\begin{align}
\partial \bar \el(\theta) &= \int \el(x) \partial g_\theta(x) P_0(dx).
\label{hprime}
\end{align}
\end{alemma}
\begin{proof}[Proof of Lemma~\ref{lem_g}]
Let
\begin{align}
\gamma_\theta(x) &\equiv 1 + \tanh[\theta a_\mu(x)]
\end{align}
be the numerator of $g_\theta(x)$.  $\gamma_\theta(x)$ is obviously
continuous with respect to $\theta$ and $x$. For all finite $\theta$
and $x$,
\begin{align}
0 < \gamma_\theta(x) < 2.
\end{align}
By the extreme value theorem \cite[Theorem~4.16]{rudin_baby}, there
exist $\theta_1,\theta_2 \in [-c,c]$ and
$x_1,x_2 \in [-\Delta,\Delta]$ such that
\begin{align}
\inf_{|\theta| \le c, |x| \le \Delta} \gamma_\theta(x) &= \gamma_{\theta_1}(x_1),
\\
\sup_{|\theta|\le c, |x| \le \Delta} \gamma_\theta(x) &= \gamma_{\theta_2}(x_2).
\end{align}
It follows that
\begin{align}
0 < c_3 \equiv \frac{\gamma_{\theta_1}(x_1)}{\gamma_{\theta_2}(x_2)}
\le g_\theta(x) \le \frac{\gamma_{\theta_2}(x_2)}{\gamma_{\theta_1}(x_1)}
\equiv c_4 < \infty.
\end{align}
The proof of Eq.~(\ref{c5}) is similar, by noting that
$\partial\gamma_\theta(x)$ is also continuous and bounded.
Equation~(\ref{c5}) then implies Eq.~(\ref{hprime}) by a corollary of
the dominated convergence theorem \cite[Corollary~2.8.7]{bogachev}.

\end{proof}

Next, I define what it means exactly for a quantity to be a function
of $\Delta$ and also normalize the $H, A, \chol,\dot H$ matrices to
remove their dependence on $\Delta$, for later use.

\begin{definition}
\label{def_R}
Let a function $T_\Delta:\mathbb R\to\mathbb R$ be
\begin{align}
T_\Delta(x) \equiv \frac{x}{\Delta}.
\end{align}
Define a standard measure $R_\theta$ on $(\mathbb R,\Sigma)$
to describe the distribution of the object with standard size as
\begin{align}
R_\theta(\cdot) &= P_\theta\Bk{T_\Delta^{-1}\bk{\cdot}},
\end{align}
such that
\begin{align}
\sup_{y \in \supp R_\theta} |y| &= 1.
\label{supp_standard}
\end{align}
As $T_\Delta$ is invertible, 
\begin{align}
P_\theta(\cdot) = R_\theta\Bk{T_\Delta\bk{\cdot}}.
\label{t_Delta}
\end{align}
A function of the true $P_0$ is said to be a function of $\Delta$ if
$R_0$ is fixed while $P_0$ varies with $\Delta$ through
Eq.~(\ref{t_Delta}).

The Hankel matrix for $R_\theta$ is
\begin{align}
G_{pq}(\theta) &\equiv \Avg{y^p,y^q}_{R_\theta}=
 \Delta^{-p-q} H_{pq}(\theta).
\end{align}
The orthonormal polynomials with respect to $R_\theta$,
satisfying $\avg{b_n,b_m}_{R_\theta} = \delta_{nm}$, are given by
\begin{align}
b_{n}(y) &= \sum_{p=0}^n B_{np}(\theta) y^p,
&
B_{np}(\theta) &= A_{np}(\theta)\Delta^p.
\end{align}
Let $V(\theta)$ be the Cholesky factor of $G(\theta)$. It is given by
\begin{align}
V_{pn}(\theta) &= \Avg{y^p,b_n}_{R_\theta} = \Delta^{-p}\chol_{pn}(\theta).
\label{V}
\end{align}
At $\theta = 0$, assuming the score given by Eq.~(\ref{a_mu}),
$\dot G_{pq}$ is given by
\begin{align}
\dot G_{pq} &= \Avg{y^{p+q},b_\mu}_{R_0} = V_{p+q\,\mu}(0).
\label{dot_G}
\end{align}
\end{definition}

I need another lemma to evaluate the derivative of the purification
given by Eq.~(\ref{Phi_theta}).
\begin{alemma}
\label{lem_diff}
Assume the conditions for Theorem~\ref{thm_qbound} and consider the
submodel given by Eqs.~(\ref{P_theta})--(\ref{theta}), (\ref{a_mu}),
(\ref{Phi_theta}), and (\ref{H_theta}).  Let
\begin{align}
\ket{\Phi^j_\theta} &\equiv
\sum_{p=0}^j \sum_{n=0}^p \frac{(-i\hat k)^p}{p!}
\chol_{pn}(\theta)\ket{\psi}\otimes\ket{n} 
\label{PhiN_theta}
\end{align}
be a function of $\theta \in (-c,c)$ with values in
$\mathcal H\otimes\mathcal H''$.
\begin{enumerate}
\item $\{\ket{\Phi^j_\theta}\}$ converges uniformly to 
  the $\ket{\Phi_\theta}$ given by Eq.~(\ref{Phi_theta}),
in the sense that
\begin{align}
\lim_{j\to\infty}
\sup_{\theta \in (-c,c)} \Norm{\ket{\Phi^j_\theta}- \ket{\Phi_\theta}} &= 0.
\end{align}

\item For any $j$ and $\theta$, the derivative of
  Eq.~(\ref{PhiN_theta}) with respect to $\theta$ exists 
in $\mathcal H\otimes\mathcal H''$ and is given
  by
\begin{align}
\partial\ket{\Phi^j_\theta} &= 
\sum_{p=0}^j \sum_{n=0}^p \frac{(-i\hat k)^p}{p!}
\partial\chol_{pn}(\theta)\ket{\psi}\otimes\ket{n}.
\label{PhiN_prime}
\end{align}
\item $\partial\ket{\Phi_\theta}$ exists in
  $\mathcal H\otimes\mathcal H''$, and
  $\{\partial\ket{\Phi^j_\theta}\}$ converges uniformly to it.
\end{enumerate}
\end{alemma}
\begin{proof}[Proof of Lemma~\ref{lem_diff}]
  Since each measure in the Szeg\H{o} class has an infinite and
  bounded support by definition, $P_0$ satisfies the condition for
  Lemma~\ref{lem_pur} by assumption. By Lemmas~\ref{lem_support} and
  \ref{lem_g}, $\supp P_\theta = \supp P_0$, so Lemma~\ref{lem_pur}
  can be applied to the whole submodel, and Eqs.~(\ref{Phi_theta}) and
  (\ref{PhiN_theta}) are well defined. To prove statement 2 of the
  lemma, note that $\partial L_{pn}$ is given by the formula
  \cite[Theorem~A.1]{sarkka}
\begin{align}
\partial\chol_{pn} &= \Delta^p \partial V_{pn},
\label{dV}
\\
\partial V_{pn} &= \sum_{m=0}^p \sum_{q=0}^m \sum_{r=0}^n V_{pm}
W_{mn} B_{mq} (\partial G_{qr}) B_{nr},
\label{sarkka}
\\
W_{mn} &\equiv \begin{cases}1, & n < m,\\
1/2, & n = m,\\
0, & n > m,\end{cases}
\end{align}
in terms of the normalized quantities in Definition~\ref{def_R}.  By
Lemma~\ref{lem_g} and Eq.~(\ref{Delta}),
\begin{align}
|\partial G_{qr}| &= \abs{\int \bk{\frac{x}{\Delta}}^{q+r}\partial g_\theta(x) P_0(dx)}
\le c_5,
\label{dG}
\end{align}
so $|\partial \chol_{pn}| < \infty$, $\partial\ket{\Phi_\theta^j}$
exists, and the basic rules of differentiation give
Eq.~(\ref{PhiN_prime}) \cite[(1.1.3) and (1.1.4)]{flett}.

To prove statements 1 and 3, the plan is to prove that the sequence of
functions satisfy the two conditions for Theorem~(1.7.1) in
Ref.~\cite{flett}, which implies the statements. The first condition
is that $\{\ket{\Phi_\theta^j}\}$ converges to $\ket{\Phi_\theta}$ at
a point $\theta$. This condition is satisfied, as
Lemmas~\ref{lem_support}, \ref{lem_g}, and \ref{lem_pur} imply the
pointwise convergence at any $\theta$. The second condition is the
uniform convergence of $\{\partial\ket{\Phi_\theta^j}\}$, which I
prove next. Suppose $J > j$ and consider
\begin{align}
&\quad \Norm{\partial\ket{\Phi^J_\theta}  - \partial\ket{\Phi^j_\theta} }^2
\label{cauchy_Phi}
\\
&= 
\Norm{\sum_{p=j+1}^J \sum_{n=0}^J
\frac{(-i\hat k)^p}{p!}\partial\chol_{pn}\ket{\psi}\otimes\ket{n}}^2
\label{step_cauchy}
\\
&= \sum_{n=0}^J
\sum_{q=j+1}^J \sum_{p=j+1}^J \frac{i^q(-i)^p}{q!p!}
\bra{\psi}\hat k^{q+p}\ket{\psi}(\partial\chol_{qn})(\partial\chol_{pn})
\label{step_series}
\\
&\le 
\sum_{n=0}^J \bk{\sum_{p=j+1}^J \frac{\opnorm{\hat k}^p}{p!} \partial\chol_{pn}}^2
\label{step_opnorm}
\\
&\le 
\sum_{n=0}^J \Bk{\sum_{p=j+1}^J \frac{(\opnorm{\hat k}\Delta)^{2p}}{p!}}
\sum_{p=j+1}^J \frac{(\partial V_{pn})^2}{p!}
\label{step_csi}
\\
&= 
\underbrace{\Bk{\sum_{p=j+1}^J \frac{(\opnorm{\hat k}\Delta)^{2p}}{p!}}}
_{s_1^J - s_1^j}
\underbrace{\sum_{p=j+1}^J \frac{1}{p!}\sum_{n=0}^p \bk{\partial V_{pn}}^2}
_{s_2^J(\theta)-s_2^j(\theta)},
\label{cauchy_s}
\end{align}
where Eq.~(\ref{step_cauchy}) has used the lower-triangularity of
$\partial\chol_{pn}$ to replace $\sum_{n=0}^p$ by $\sum_{n=0}^J$,
Eq.~(\ref{step_opnorm}) has used the series of absolute values to
bound Eq.~(\ref{step_series}) and the definition of the operator norm
to bound
$|\bra{\psi}\hat k^{q+p}\ket{\psi}| \le \opnorm{\hat k}^{q+p}$,
Eq.~(\ref{step_csi}) has used the CSI and Eq.~(\ref{dV}),
and $s_1^j$ and $s_2^j(\theta)$ are defined as
\begin{align}
s_1^j &\equiv \sum_{p=0}^j \frac{(\opnorm{\hat k}\Delta)^{2p}}{p!},
\\
s_2^j(\theta) &\equiv \sum_{p=0}^j \frac{1}{p!}\sum_{n=0}^p [\partial V_{pn}(\theta)]^2.
\label{s2}
\end{align}
If $\{s_1^j\}$ converges and $\{s_2^j(\theta)\}$ converges uniformly,
then Eq.~(\ref{cauchy_Phi}) can be bounded uniformly. $\{s_1^j\}$
converges to $\exp[(\opnorm{\hat k}\Delta)^2]$ for a bounded $\hat k$,
so it remains to be proved that $\{s_2^j(\theta)\}$ converges
uniformly.

To bound $(\partial V_{pn})^2$, apply the CSI to 
Eq.~(\ref{sarkka}) to obtain
\begin{align}
(\partial V_{pn})^2 &\le \Bk{\sum_{m=0}^p (V_{pm})^2}
\sum_{m=0}^p \bk{W_{mn} D_{mn}}^2,
\\
D &\equiv B_{(p)} [\partial G_{(p)}] B_{(p)}^\top,
\\
\sum_{n=0}^p (\partial V_{pn})^2  &\le G_{pp} \sum_{n=0}^p \sum_{m=0}^p 
\bk{W_{mn}D_{mn}}^2
\label{step_Hpp}
\\
&\le \frac{G_{pp}}{2} \fnorm{D}^2 \le \frac{1}{2} \fnorm{D}^2,
\label{step_HS}
\end{align}
where $G_{pp} = \sum_{m=0}^p(V_{pm})^2$ in Eq.~(\ref{step_Hpp})
comes from the fact that $V$ is the Cholesky factor of $G$,
$\fnorm{\cdot}$ is the Hilbert-Schmidt norm (also called the Frobenius
norm)
\begin{align}
\fnorm{D}^2 &\equiv \sum_{m,n} (D_{mn})^2,
\end{align}
the first bound in Eq.~(\ref{step_HS}) comes from the definition of
$W$ and the symmetry of $D$ \cite[Eq.~(B34)]{tsang19}, and the last
bound in Eq.~(\ref{step_HS}) comes from
\begin{align}
G_{pp} &= \norm{y^p}_{R_\theta}^2 \le 1,
\end{align}
by Eq.~(\ref{supp_standard}). $\fnorm{D}$ can be further bounded in
terms of $\fnorm{\partial G_{(p)}}$ and the operator norm
$\opnorm{B_{(p)}}$ as \cite[Theorem~II.2.11 and
Theorem~II.3.9]{stewart}
\begin{align}
\fnorm{D}  &\le \fnorm{\partial G_{(p)}} \opnorm{B_{(p)}}^2
= \frac{\fnorm{\partial G_{(p)}}}{\lambda_{\mathrm{min}}[G_{(p)}]},
\label{norm_D}
\end{align}
where $\lambda_{\mathrm{min}}$ denotes the smallest eigenvalue and the
last step has used $B_{(p)}^\top B_{(p)} = G_{(p)}^{-1}$ from
Lemma~\ref{lem_poly} and Eq.~(II.2.10) in Ref.~\cite{stewart} to
obtain
\begin{align}
\opnorm{B_{(p)}}^2 &= \lambda_{\mathrm{max}}[G_{(p)}^{-1}] = 
\frac{1}{\lambda_{\mathrm{min}}[G_{(p)}]}.
\end{align}
To bound $\lambda_{\mathrm{min}}[G_{(p)}]$, first note that, for any column
vector $u \in \mathbb R^p$,
\begin{align}
\el^\top G_{(p)}(\theta) \el &= 
\int \Bk{\sum_{q = 0}^p \el_q \bk{\frac{x}{\Delta}}^q}^2 g_\theta(x) P_0(dx)
\\
&\ge c_3 \el^\top G_{(p)}(0) \el,
\label{step_c3}
\end{align}
where $c_3 > 0$ comes from Lemma~\ref{lem_g}. Applying the Rayleigh
quotient theorem \cite[Theorem~4.2.2]{horn}
\begin{align}
\lambda_{\mathrm{min}}[G_{(p)}(\theta)] &= 
\min_{\el^\top \el = 1} \el^\top G_{(p)}(\theta) \el
\end{align}
to Eq.~(\ref{step_c3}), I obtain
\begin{align}
\lambda_{\mathrm{min}}[G_{(p)}(\theta)] &\ge c_3 \lambda_{\mathrm{min}}[G_{(p)}(0)].
\end{align}
If $P_0$ is in the Szeg\H{o} class, $R_0$ is also in the Szeg\H{o}
class, and by a theorem of Widom and Wilf \cite{widom66} there exists
a constant $0 < r < 1$ such that $\lambda_{\mathrm{min}}$ of its
Hankel matrix satisfies
\begin{align}
\lambda_{\mathrm{min}}[G_{(p)}(0)] &=  \Theta(\sqrt{p}r^p),
\label{eigenvalue_szego}
\end{align}
where the $\Theta$ order here is in terms of $p \to \infty$.  To bound
the other quantity $\fnorm{\partial G_{(p)}}$ in Eq.~(\ref{norm_D}), use
Eq.~(\ref{dG}) to write
\begin{align}
\fnorm{\partial G_{(p)}}^2 &\le
\sum_{n=0}^p \sum_{m=0}^p c_5^2  = c_5^2 (p+1)^2.
\end{align}
Putting everything together, $s_2^j(\theta)$ in Eq.~(\ref{s2}) can be
bounded as
\begin{align}
s_2^j(\theta) &\le 
\sum_{p=0}^j 
\frac{c_5^2(p+1)^2}{p!2 [c_3\Theta(\sqrt{p}r^p)]^2}.
\label{step_final_sum}
\end{align}
The right-hand side does not depend on $\theta$ and converges by the
ratio test, so $\{s_2^j(\theta)\}$ converges uniformly by the
Weierstrass test \cite[Theorem~7.10]{rudin_baby}.

With the convergence of $\{s_1^j\}$ and the uniform convergence of
$\{s_2^j(\theta)\}$, given any $\varepsilon > 0$, there exists a
$\theta$-independent $j_0$ such that, for all $\theta \in (-c,c)$ and
$J\ge j \ge j_0$, $s_1^J - s_1^j < \varepsilon$,
$s_2^J(\theta)-s_2^j(\theta) < \varepsilon$, and by
Eqs.~(\ref{cauchy_Phi})--(\ref{cauchy_s}),
\begin{align}
\Norm{\partial\ket{\Phi^J_\theta}  - \partial\ket{\Phi^j_\theta} } &< \varepsilon,
\label{cauchy}
\\
\sup_{\theta \in (-c,c)}
\Norm{\partial\ket{\Phi^J_\theta}  - \partial\ket{\Phi^j_\theta} } &\le \varepsilon,
\end{align}
which means that $\{\partial\ket{\Phi_\theta^j}\}$ is Cauchy in the
space of bounded Hilbert-space-valued functions
$\mathcal B_{\mathcal H\otimes \mathcal H''}[(-c,c)]$ with the
supremum norm $\sup_{\theta \in (-c,c)}\norm{\ket{\el_\theta}}$ for
each
$\ket{\el_\theta} \in \mathcal B_{\mathcal H\otimes \mathcal
  H''}[(-c,c)]$ \cite[Sec.~7.1]{dieudonne}.  As the Hilbert space
$\mathcal H\otimes\mathcal H''$ is a Banach space,
$\mathcal B_{\mathcal H\otimes \mathcal H''}[(-c,c)]$ is also a Banach
space \cite[(7.1.3)]{dieudonne}, the completeness of which means that
the Cauchy $\{\partial\ket{\Phi^j_\theta}\}$ converges
uniformly. Hence, the two conditions for Theorem~(1.7.1) in
Ref.~\cite{flett} are satisfied, and the theorem implies statements 1
and 3 of the lemma here.
\end{proof}

\begin{proof}[Proof of Theorem~\ref{thm_qbound}] Consider the submodel
  given by Eqs.~(\ref{P_theta})--(\ref{theta}), (\ref{a_mu}),
  (\ref{Phi_theta}), and (\ref{H_theta}). Let
  $\ket{\dot\Phi} \equiv \partial\ket{\Phi_\theta}|_{\theta = 0}$. The
  Helstrom information of the purified model is the Fubini-Study
  metric \cite{fujiwara95}
\begin{align}
\norm{S_\Phi}^2_{\Phi_0} &=  
4 \bk{\braket{\dot\Phi|\dot\Phi} - |\braket{\dot\Phi|\Phi_0}|^2}
\le 4 \braket{\dot\Phi|\dot\Phi}.
\label{loose_bound}
\end{align}
Lemma~\ref{lem_diff} implies that, under the conditions for
Theorem~\ref{thm_qbound}, $\braket{\dot\Phi|\dot\Phi} < \infty$,
$\{\ket{\dot\Phi^j} \equiv \partial\ket{\Phi^j_\theta}|_{\theta =
  0}\}$ converges to $\ket{\dot\Phi}$, and therefore the squared norm
also converges to \cite[Eq.~(3.9)]{debnath05}
\begin{align}
\braket{\dot\Phi^j|\dot\Phi^j} &\to \braket{\dot\Phi|\dot\Phi} < \infty.
\label{helstrom_limit}
\end{align}
With Eqs.~(\ref{k_spectral}), (\ref{Q}), and (\ref{PhiN_prime}), I can
write
\begin{align}
\braket{\dot\Phi^j|\dot\Phi^j}
= \int \sum_{n=0}^j \abs{\sum_{p=0}^j \frac{(-ik)^p}{p!}\dot \chol_{pn}}^2 Q(dk).
\label{dotPhiN}
\end{align}
To evaluate this expression as a function of $\Delta$, write the
formula for $\dot L_{pn}$ in Eqs.~(\ref{dV}) and (\ref{sarkka}) at
$\theta = 0$ as
\begin{align}
\dot\chol_{pn} &= \Delta^p \sum_{m=0}^p \sum_{q=0}^m \sum_{r=0}^n V_{pm}
W_{mn} B_{mq} V_{q+r\,\mu} B_{nr},
\label{sarkka2}
\end{align}
where $\dot G_{qr} = V_{q+r\,\mu}$ comes from Eq.~(\ref{dot_G}).  As
$V, W, B$ are all lower-triangular, for $\dot\chol_{pn}$ to be nonzero,
the indices in Eq.~(\ref{sarkka2}) should satisfy
\begin{align}
p &\ge m \ge n \ge r,
&
m &\ge q, & q + r \ge \mu,
\label{indices}
\end{align}
which imply $2m \ge q+m \ge q+r \ge \mu$ and
\begin{align}
p&\ge m \ge \ceil{\frac{\mu}{2}}.
\end{align}
Thus,
\begin{align}
\dot\chol_{pn} &= 0 \quad \textrm{if }p < \ceil{\frac{\mu}{2}}.
\label{dotL}
\end{align}
With $\dot \chol_{pn} \propto \Delta^p$, Eq.~(\ref{dotPhiN}) is a
power series of $\Delta$.  By Eq.~(\ref{helstrom_limit}), the power
series converges, so I can conclude from Eqs.~(\ref{helstrom_limit}),
(\ref{dotPhiN}), and (\ref{dotL}) that \cite{miller06}
\begin{align}
\braket{\dot\Phi|\dot\Phi} &= O(\Delta^{2\ceil{\mu/2}}).
\label{gamma_order}
\end{align}
To evaluate the Helstrom bound, I also need $\dot\beta$.
Given Eqs.~(\ref{inner_beta}), (\ref{a_mu}),
and (\ref{b_gen}), it is given by
\begin{align}
\dot\beta &= \Avg{a_\mu,x^\mu+o(\Delta^\mu)}_{P_0} =
\Delta^\mu V_{\mu\mu}+\Avg{a_\mu,o(\Delta^\mu)}_{P_0}
\nonumber\\
&= \Theta(\Delta^\mu),
\label{dot_beta_order}
\end{align}
where the second equality comes from Lemma~\ref{lem_poly} and
Eq.~(\ref{V}) and the final equality comes from $V_{\mu\mu} > 0$ by
Lemma~\ref{lem_poly} and the CSI
\begin{align}
\abs{\Avg{a_\mu,o(\Delta^\mu)}_{P_0}} \le 
\norm{a_\mu}_{P_0} \norm{o(\Delta^\mu)}_{P_0}= o(\Delta^\mu).
\end{align}
Finally, the theorem is given by
\begin{align}
\tilde{\mathsf H} &\ge \mathsf H \ge \frac{\dot\beta^2}{N\norm{S_\Phi}_{\Phi_0}^2}
\ge \frac{\dot\beta^2}{4N\braket{\dot\Phi|\dot\Phi}} 
= \frac{\Omega(\Delta^{2\floor{\mu/2}})}{N},
\end{align}
where the first inequality comes from the submodel bound in
Lemma~\ref{lem_submodel}, the second inequality comes from applying
the purification bound in Lemma~\ref{lem_monotone} to the purification
in Lemma~\ref{lem_pur}, the third inequality comes from
Eq.~(\ref{loose_bound}), and the final equality comes from
Eqs.~(\ref{gamma_order}) and (\ref{dot_beta_order}).
\end{proof}

\begin{remark}
  Although the assumption of an infinite and bounded $\supp P_0$ is
  central to Lemma~\ref{lem_pur}, Lemma~\ref{lem_g}, and
  Theorem~\ref{thm_qbound}, the more specific Szeg\H{o}-class
  assumption is used only in Eqs.~(\ref{eigenvalue_szego}) and
  (\ref{step_final_sum}) in Lemma~\ref{lem_diff}, to ensure the
  convergence of the Helstrom information of the purified model. If
  $\lambda_{\mathrm{min}}[G_{(p)}(0)]$ can be lower-bounded in another
  way that still ensures the uniform convergence of
  $\{s_2^j(\theta)\}$ in Eq.~(\ref{step_final_sum}), then
  Lemma~\ref{lem_diff} and Theorem~\ref{thm_qbound} still hold, and
  the Szeg\H{o}-class assumption may be relaxed. For example,
  Ref.~\cite[Theorem~3]{widom67} gives the asymptotic behavior of
  $\lambda_{\mathrm{min}}[G_{(p)}(0)]$ for a more general class of
  $P_0$ that still makes $\{s_2^j(\theta)\}$ converge uniformly,
  although the conditions there are much harder to state or check.
\end{remark}

\section{\label{app_prop_direct}Proof of
  Proposition~\ref{prop_direct}}
\begin{proof}[Proof of Proposition~\ref{prop_direct}] Let $\nu[P]$ be
  the probability measure for each photon, with a probability density
  given by Eq.~(\ref{eta}), and let
  $\{\nu[P_\theta]:\theta \in (-c,c)\}$ be the 1D submodel based on
  the $\{P_\theta\}$ given by Eqs.~(\ref{P_theta})--(\ref{theta}) and
  (\ref{a_mu}). The score $\nu_*S$ of $\{\nu[P_\theta]\}$ is given by
\begin{align}
(\nu_*S)(\xi) &= \frac{\dot\eta(\xi)}{\eta_0(\xi)},
\label{score_direct}
\\
\eta_\theta(\xi) &\equiv \eta(\xi,P_\theta) = \int h(\xi-x) P_\theta(dx),
\\
\dot\eta(\xi) &= \int h(\xi-x) a_\mu(x) P_0(dx).
\label{dot_eta}
\end{align}
Assume the normalized quantities defined in Definition~\ref{def_R} in
the following.  Using Taylor's theorem \cite{rudin_baby},
$h(\xi-x) = h(\xi-\Delta y)$ becomes
\begin{align}
h(\xi-\Delta y)
&= \sum_{p=0}^\mu \frac{h^{(p)}(\xi)}{p!}(-\Delta y)^p
+ r(\xi,y) \Delta^{\mu+1},
\end{align}
where the remainder is given by
\begin{align}
r(\xi,y) &\equiv \frac{h^{(\mu+1)}(\xi-\Delta \bar y)}{(\mu+1)!}
(-y)^{\mu+1},
\label{remainder}
\end{align}
for a certain $\bar y(\xi,y)$ between $0$ and
$y$. Equation~(\ref{dot_eta}) becomes
\begin{align}
\dot\eta(\xi) &= 
\frac{h^{(\mu)}(\xi)}{\mu!}(-\Delta)^\mu V_{\mu\mu}
+ \Delta^{\mu+1} \Avg{r(\xi,y),b_\mu}_R.
\label{dot_eta2}
\end{align}
The last term can be bounded as
\begin{align}
\abs{\Avg{r(\xi,y),b_\mu}_R} &\le \norm{r(\xi,y)}_R
\label{step_csi_rb}
\\
&\le \frac{\bar{h}(\xi)}{(\mu+1)!}\norm{y^{\mu+1}}_R
\label{step_hmu1}
\\
&\le \frac{\bar{h}(\xi)}{(\mu+1)!},
\label{step_suppR}
\end{align}
where Eq.~(\ref{step_csi_rb}) has used the CSI, Eq.~(\ref{step_hmu1})
has used Eq.~(\ref{h_mu1}) to bound the
$|h^{(\mu+1)}(\xi-\Delta \bar y)|$ in Eq.~(\ref{remainder}), and
Eq.~(\ref{step_suppR}) comes from Eq.~(\ref{supp_standard}). Apply the
triangle inequality to Eq.~(\ref{dot_eta2}) to obtain
\begin{align}
\abs{\dot\eta(\xi)} &\le
\abs{\frac{h^{(\mu)}(\xi)}{\mu!}(-\Delta)^\mu V_{\mu\mu}}
+ \abs{\Delta^{\mu+1} \Avg{r(\xi,y),b_\mu}_R}
\\
&\le \frac{\Delta^\mu}{\mu!} 
\Bk{V_{\mu\mu}\abs{h^{(\mu)}(\xi)}  + \frac{\Delta \bar{h}(\xi)}{\mu+1}}.
\end{align}
The norm of the score given by Eq.~(\ref{score_direct}) can be bounded
again by the triangle inequality as
\begin{align}
\Norm{\frac{\dot\eta}{\eta_0}}_{\nu[P_0]}
&\le \frac{\Delta^\mu}{\mu!}
\left[V_{\mu\mu} \Norm{\frac{h^{(\mu)}}{\eta_0}}_{\nu[P_0]} \right.
\nonumber\\&\qquad\quad
\left.
+\frac{\Delta}{\mu+1}\Norm{\frac{\bar{h}}{\eta_0}}_{\nu[P_0]}\right].
\label{score_bound}
\end{align}
Equation~(\ref{h0}) implies that
\begin{align}
\eta_0(\xi) \ge \underline{h}(\xi),
\end{align}
while Eqs.~(\ref{finite_int}) imply that
\begin{align}
\Norm{\frac{h^{(\mu)}}{\eta_0}}_{\nu[P_0]}^2 &\le 
\int \frac{[h^{(\mu)}(\xi)]^2}{\underline{h}(\xi)} d\xi < \infty,
\\
\Norm{\frac{\bar{h}}{\eta_0}}_{\nu[P_0]}^2
&\le 
\int \frac{[\bar{h}(\xi)]^2}{\underline{h}(\xi)} d\xi < \infty.
\end{align}
The convergence of these quantities means that Eq.~(\ref{score_bound})
can be written as
\begin{align}
\Norm{\frac{\dot\eta}{\eta_0}}_{\nu[P_0]} &= O(\Delta^\mu),
\end{align}
and the Fisher information of the submodel becomes
\begin{align}
&\quad 
\norm{\nu_* S}_{\nu[P_0]}^2 
= \Norm{\frac{\dot\eta}{\eta_0}}_{\nu[P_0]}^2 =O(\Delta^{2\mu}).
\end{align}
Given this expression and Eq.~(\ref{dot_beta_order}), the Cram\'er-Rao
bound for the $M$-temporal-mode submodel
$\{\nu^{(M,\epsilon)}[P_\theta]:\theta\in (-c,c)\}$ is
\begin{align}
\mathsf C &= \frac{\dot\beta^2}{N\norm{\nu_* S}_{\nu[P_0]}^2}
= \frac{\Omega(1)}{N},
\end{align}
and the fact that any submodel bound is a lower bound on the
semiparametric bound $\tilde{\mathsf C}$ \cite[Lemma~25.19]{vaart}
leads to Eq.~(\ref{CRB_direct_bound}).
\end{proof}

\bibliography{research2}

\end{document}